\documentclass[JHEP,12pt]{article}
\usepackage{jheppub}
\usepackage{amssymb,amsmath,amsthm}
\usepackage[shortlabels]{enumitem}
\usepackage{braket}
\usepackage{bbold}
\usepackage{comment} 
\usepackage{graphicx, xcolor, varwidth, array}
\usepackage{tikz,tikz-3dplot}
\usepackage{color}
\usepackage{float}
\newtheorem{thm}{Theorem}
\newtheorem{conj}[thm]{Conjecture}
\theoremstyle{remark}
\newtheorem{defn}[thm]{Definition}
\newtheorem{lem}[thm]{Lemma}
\newtheorem{rem}[thm]{Remark}
\newtheorem{cor}[thm]{Corollary} 
\newtheorem{conv}[thm]{Convention}

\newcommand{\deltabar}{\delta\hspace*{-0.2em}\bar{}\hspace*{0.2em}}

\DeclareMathOperator{\cl }{cl}
\DeclareMathOperator{\setint }{int}
\DeclareMathOperator{\A}{Area}

\renewcommand{\S}{S_{\rm gen}}
\newcommand{\hmax}{H_{\rm max}}
\newcommand{\hmin}{H_{\rm min}}
\newcommand{\emax}{e_{\rm max}}
\newcommand{\emin}{e_{\rm min}}

\newcommand{\smax}{\Sigma_{\rm max}}
\newcommand{\smin}{\Sigma_{\rm min}'}

\usepackage{tablefootnote}
\usepackage{caption}

\newcommand{\hmingen}{H_{\rm min,gen}}
\newcommand{\hmg}{H_{\rm max,gen}}
\newcommand{\thmax}{\Theta_{\rm max}}

\author{Raphael Bousso and Elisa Tabor}

\affiliation{Center for Theoretical Physics and Department of Physics,\\
University of California, Berkeley, California 94720, U.S.A. 
} 

\emailAdd{bousso@berkeley.edu}
\emailAdd{etabor@berkeley.edu}

\title{Discrete Max-Focusing}

\abstract{%
The Quantum Focusing Conjecture (QFC) lies at the foundation of holo\-graphy and semiclassical gravity. The QFC implies the Bousso bound and the Quantum Null Energy Condition (QNEC). The QFC also ensures the consistency of the quantum extremal surface prescription and bulk reconstruction in AdS/CFT. 

However, the central object in the QFC --- the expansion of lightrays --- is not defined at points where geodesics enter or leave a null congruence. Moreover, the expansion admits three inequivalent quantum extensions in terms of the conditional max, min, and von Neumann entropies. 

Here we formulate a discrete notion of nonexpansion that can be evaluated even at non-smooth points. Moreover, we show that a single conjecture, the discrete \emph{max}-QFC, suffices for deriving the QNEC, the Bousso bound, and key properties of both max and min entanglement wedges. Continuous numerical values need not be assigned, nor are the von Neumann or min-versions of the quantum expansion needed. Both our new notion of nonexpansion, and also the properties of conditional max entropies, are inherently asymmetric and outward directed from the input wedge. Thus the framework we develop here reduces and clarifies the axiomatic structure of semiclassical gravity, eliminating redundancies and fixing ambiguities.

We also derive a new result: the strong subadditivity of the generalized smooth conditional max and min entropies of entanglement wedges.}

\makeatletter
\gdef\@fpheader{\mbox{}}
\makeatother

\begin{document}
\maketitle

\section{Introduction}

The Covariant Entropy Bound~\cite{Bousso:1999cb,Bousso:1999xy} and its generalization, the Quantum Focusing Conjecture (QFC)~\cite{Bousso:2015mna,Shahbazi-Moghaddam:2022hbw}, capture a universal relation between quantum information and spacetime geometry. The QFC is critical to our understanding of how spacetime emerges in the Anti-de Sitter/Conformal Field Theory (AdS/CFT) correspondence~\cite{Maldacena:1997re}. The QFC underlies vital properties of entanglement wedges~\cite{Ryu:2006bv,Ryu:2006ef,Hubeny:2007xt,Wall:2012uf,Faulkner:2013ana,Engelhardt:2014gca,Bousso:2022hlz,Bousso:2023sya} such as nesting, complementarity, and strong subadditivity of their generalized entropy.

The Covariant Entropy Bound and the QFC are not tied to AdS: they apply in arbitrary spacetimes. Entanglement wedges, too, can be generalized~\cite{Bousso:2022hlz,Bousso:2023sya}
. An entanglement wedge can be associated to an arbitrary gravitating input region, analogous to CFT subregions on the boundary of AdS, in any spacetime. The QFC ensures that generalized entanglement wedges satisfy all of the above properties, suggesting that they represent holographically reconstructible regions in arbitrary spacetimes. Thus, the QFC is likely to play a key role in the search for a quantum theory of gravity that can describe our own universe.

The QFC implies~\cite{Bousso:2015mna} a quantum refinement of the Covariant Entropy Bound, which in turn implies~\cite{Flanagan:1999jp} the Generalized Second Law of Thermodynamics for causal horizons~\cite{Bekenstein:1972tm,Jacobson:2003wv}. The QFC also implies a novel result in quantum field theory without gravity, the Quantum Null Energy Condition (QNEC)~\cite{Bousso:2015mna}. The QNEC was later proven (much more laboriously), using only quantum field theory itself~\cite{Bousso:2015wca,Balakrishnan:2017bjg,Ceyhan:2018zfg}. 

Let us briefly review the statement of the QFC. In spacetimes satisfying Einstein's equation with the stress tensor obeying the Null Energy Condition (NEC), classical focusing holds. That is, the geometric expansion $\theta$ of a congruence of null geodesics cannot increase as a function of affine parameter along each geodesic: $\theta'\leq 0$. However, the NEC is violated in valid states of relativistic field theories such as the standard model, so classical focusing is false. 

The QFC is a semiclassical extension of classical focusing which does appear to hold in all known examples. The QFC arises from classical focusing by replacing the area of surfaces with a quantum-corrected area, the generalized entropy. Given a wedge $a$ (the domain of dependence of a spatial region), the generalized entropy is given by $\S(a)=\A(a)/(4G)+S_{\rm vN}(a)+\ldots$, where $\A(a)$ is the area of the edge of $a$ and $S_{\rm vN}$ is the von Neumann entropy of the quantum fields restricted to $a$. One can consider the behavior of $\S$ under outward deformations of the wedge $a$ in an orthogonal future null direction. The functional derivative of $\S$ defines a future quantum expansion, $\Theta_{\rm vN}^+(a,p)$ at every point $p$ on the edge of the wedge. The past quantum expansion is defined analogously. The QFC is the statement that the quantum expansion cannot increase; schematically, $(\Theta_{\rm vN}^\pm)'\leq 0$. 

In this paper, we will revisit the QFC and modify its axiomatic formulation. We are motivated by the following considerations.

1. In recent years, our understanding of holography and spacetime emergence has been substantially refined. The QES prescription for the entanglement wedge of CFT subregions is valid only for certain ``compressible'' quantum states in the bulk~\cite{Akers:2020pmf}. In general states, there are two distinct wedges: a max entanglement wedge that can be fully reconstructed from the boundary with high fidelity; and a (larger) min entanglement wedge, the smallest region such that no operators in its exterior can be reconstructed with low fidelity. The definitions of these wedges involve notions from \emph{single-shot}\footnote{The task of communicating a long sequence of an independent and identically distributed random variable is called \emph{asymptotic}. In this case, compression can reduce the resources needed per message, which are quantified by the von Neumann entropy. A classical example is English text, with each letter corresponding to a message. Compression algorithms take advantage of the fact that only a small fraction of long letter sequences can actually occur, whereas most combinations have negligible probability. \emph{Single-shot} communication involves sending only one message. In general, this will be less efficient. (For example, sending just one letter of a 26-letter alphabet requires $\log_2 26$ bits, noticeably greater than the number of bits per letter for a long text.) Resources can still be saved by allowing a small classical failure probability $\epsilon$ (e.g., by betting that the letter will not be \emph{z}), or, in the quantum case, by allowing a small distance between sent and received states. This can be quantified by the max entropy.} quantum communication: the (generalized) smooth max and min conditional entropies~\cite{RenWol04a}. The QFC, however, is based on the (generalized) von Neumann entropy, which quantifies the resources needed for \emph{asymptotic} communication tasks. Recently Akers \emph{et al.}~\cite{Akers:2023fqr} have proposed a max  and min-version of the QFC, so that altogether three inequivalent conjectures are extant: max , min-, and von-Neumann-QFC. It would be preferable to have a more economical axiomatic structure that assumes only one QFC. The min-QFC implies the von-Neumann-QFC~\cite{Akers:2023fqr}. However, the max-QFC and the min-QFC are logically independent; neither implies the other. In this paper we will show that the max-QFC suffices; neither a von Neumann QFC nor a min-QFC are needed.

2. Independently of the above development, Shahbazi-Moghaddam has noted~\cite{Shahbazi-Moghaddam:2022hbw} that the original von-Neumann QFC can be somewhat weakened without affecting its usefulness in any known application. Namely, the quantum expansion may well increase; it just cannot \emph{become} positive. This ``Restricted QFC'' was proven to hold in brane-world models, whereas a proof of the original QFC in that setting does not appear to be possible. Although no explicit counterexample to the original QFC is known, we are persuaded by these observations to consider the possibility that only the restricted version holds. Indeed, the version of the QFC we will propose is even more ``bare-bones'' in that it will require only a sign of the expansion, not a numerical value.  

3. The classical expansion of orthogonal null geodesics is not well-defined at all points on the edge of a wedge. The expansion need not be continuous (for example, consider lightrays orthogonal to a flat plane that meets a cylinder segment in Minkowski space), and it is not uniquely defined at such discontinuities. The expansion is also ill-defined at conjugate points or self-intersections of a null congruence, where null geodesics enter or leave the congruence. Such points develop generically even along a null congruence that is initially smooth. These shortcomings are not fatal. But they have required us to add somewhat \emph{ad-hoc} rules~\cite{Bousso:2023sya} when we apply the QFC across such points, as we often must. Finally, the expansion is also not defined on null portions of an achronal codimension 2 surface, which do arise on the edges of wedges and can even arise on the edge of a (generalized) entanglement wedge; see Fig.~\ref{fig:crescent}. (See Sections~\ref{sec:prelim} and \ref{sec:entanglement} for definitions of these terms.) By assigning only a sign to the expansion, rather than a numerical value, we will be able to treat all of these cases satisfactorily through a single definition of the expansion.


\begin{figure}
    \centering
    \includegraphics[width=9cm]{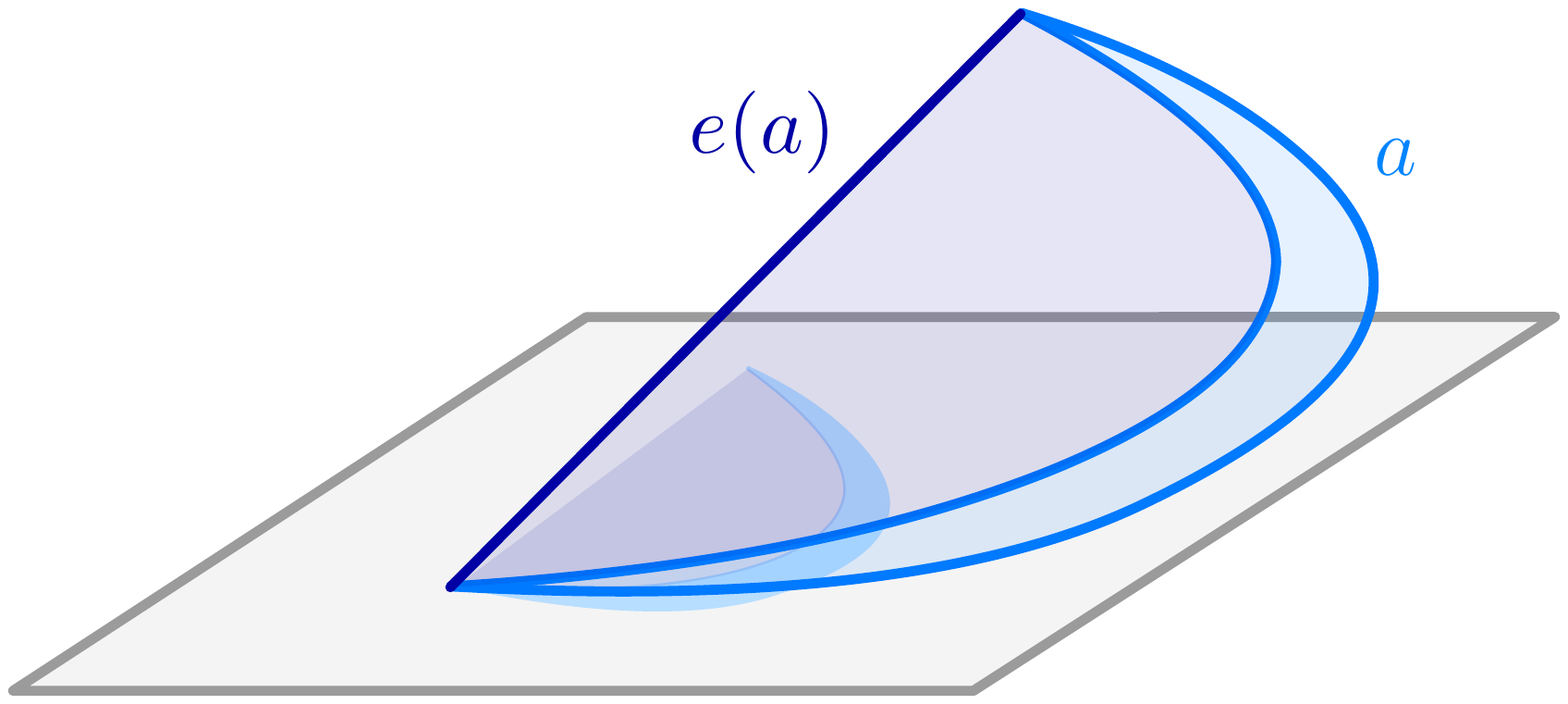}
    \caption{Generalized entanglement wedge $e(a)$ of a ``boosted crescent'' wedge $a$ in 2+1-dimensional Minkowski space. The tips of the crescent are null separated, and the edge of $e(a)$ contains the null geodesic that connects these two points. On this null portion, a numerical expansion is not defined and hence cannot be said to vanish. Thus, it is not correct to say that $e(a)$ is ``extremal'' (in the sense of vanishing expansion) on the edge portion $\eth e(a)\setminus \eth a$. However, our discrete notion of nonexpansion, Def.~\ref{def:fne}, does hold on this edge portion. Moreover, the notion of ``extremal'' can be replaced with the more precise and robust notion of ``throat,'' which also holds here; see Def.~\ref{def:throat} and Theorem~\ref{thm:emaxextremal}.}
    \label{fig:crescent}
\end{figure}

The goal of this paper is to reformulate and simplify the QFC in response to the above developments and shortcomings. To address the first point, we will show that only the max-QFC is needed in theorems. The second and third points motivate us to simplify the notion of expansion. Instead of assigning a numerical value $\Theta^+$ to the expansion, we will introduce a more robust, yet sufficient, qualitative criterion. At each point $p$ on the edge of a wedge $a$, our criterion establishes only whether or not $a$ is \emph{future-nonexpanding} at $p$ (and of course, whether or not $a$ is \emph{past-nonexpanding} at $p$). This criterion can be evaluated at ``corner points,'' at discontinuities of the classical expansion, and on null portions of surfaces, without any additional structure or \emph{ad-hoc} rules. 

Thus our reformulation minimizes the axiomatic foundation of the holographic properties of semiclassical spacetime, while maximizing its reach. An inward max-expansion is no longer defined, nor is an outward min-expansion defined. Only the outward max-expansions are explicitly defined, and not as a numerical quantity, but only through the qualitative notion of outward max-nonexpansion. (Inward min-expansion remains implicitly defined through the outward max-expansions of the complement wedge, but for conceptual clarity only the latter is referred to in this paper. The definition of the generalized min entanglement wedge as the complement of a generalized max entanglement wedge of a suitable complement region will be further developed elsewhere \cite{BoussoKaya}.) 

Our work has two strands that may seem unrelated at first. The first strand is the generalization of the classical notion of nonexpansion to nonsmooth surfaces such as those resulting from intersections and caustics of null congruences. The second strand is the reduction of the types of entropies that appear in definitions to the max conditional entropy only. We would present these results in separate manuscripts, were it not for the fact that the two strands come together in the final result that only \emph{outward} deformations of wedges should be considered. Classical nonexpansion is about what happens next as we \emph{enlarge} a wedge; our definition of ``nonexpanding'' cannot be reproduced by considering inward deformations instead. Similarly, unlike the von Neumann entropy, the conditional max entropy is also intrinsically outward-directed. This is inherent in its mathematical properties: the max and min entropies obey opposite chain rule inequalities that work well with unions and intersections, respectively. It is also inherent in its physical interpretation as measuring the resources needed to reconstruct semiclassical data in a larger region from a more fundamental algebra of a smaller one~\cite{Bousso:2022hlz,Bousso:2023sya}. It is interesting that these seemingly unrelated considerations both give preference to outward, rather than inward, deformations of wedges.

\paragraph{Outline and Summary of Results} In Sec.~\ref{sec:wedges}, we review the definition of a wedge and its relation to causal structure. In Sec.~\ref{sec:hmaxgen}, we introduce the key concept of generalized max conditional entropy $\hmg(b|a)$, a modification of the more familiar notion of the difference between the generalized entropies, $\S(b)-\S(a)$, of two nested wedges $a\subset b$. At leading order in Newton's constant $G$, $\hmg(b|a)$ is given by the difference between the areas of the wedges. At subleading order, the standard quantum mechanical smooth max entropy of $b$ conditioned on $a$ contributes to $\hmg(b|a)$. Up to these orders, $\hmg$ is known to satisfy useful properties such as strong subadditivity and certain chain rules. We review these properties, and we conjecture that they hold for the full quantity.

In Sec.~\ref{sec:discrete}, we define discrete notions of \emph{future-non\-expanding} and \emph{past-non\-expanding}, in terms of the negative sign of the generalized smooth max entropy of outward deformations of a wedge $a$, conditioned on the original $a$. Importantly, these definitions make sense at all points on $\eth a$, whether they lie on null portions, corners, or smooth spacelike portions of $\eth a$. We formulate a discrete, max-version of the Quantum Focusing Conjecture based on these definitions. In Sec.~\ref{sec:imply} we establish some immediate implications of the definitions and conjecture: the fact that null portions of an edge are past- and future-nonexpanding; the persistence of nonexpansion along lightsheets; the Bousso bound; and the Generalized Second Law of thermodynamics. %
In Sec.~\ref{sec:continuous}, we identify circumstances under which the expansion can be assigned a continuous value $\Theta^\pm$, and we prove relations between the discrete and continuous notions of nonexpansion.

The close relation to several well-tested conjectures established in Sec.~\ref{sec:imply} provides initial support for Discrete Max-Focusing. In Sec.~\ref{sec:evidence}, we provide additional evidence. We show that Discrete Max-Focusing implies a number of nontrivial results that can be rigorously proven: the fact that null geodesics can leave, but not enter, the Cauchy horizons of wedges in any Lorentzian spacetime; the classical focusing property of General Relativity in spacetimes satisfying the Null Energy Condition; and finally, the Quantum Null Energy Condition, a proven lower bound on the stress tensor in relativistic quantum field theories.

Therefore Discrete Max-Focusing implies, in the respective limits, all further consequences of the conditions listed in the previous two paragraphs. Penrose's singularity theorem~\cite{Penrose:1964wq} follows from the present work in the regime in which it is formulated (the classical limit, in which the Null Energy Condition emerges from Discrete Max-Focusing). Similarly, Wall's quantum singularity theorem~\cite{Wall:2010jtc} follows in its own regime, the limit where numerical von Neumann quantum expansions exist and agree (at least in sign) with the max-expansions. Other such statements include the area theorem~\cite{Hawking:1971vc}, the Gao-Wald theorem~\cite{Gao:2000ga}, and causal wedge inclusion for von-Neumann entanglement wedges~\cite{Wall:2012uf}. 

An important future goal will be to formulate all important classical and semiclassical gravity theorems directly in the fully general language of Discrete Max-Focusing (rather than recovering them as limits), and to adapt their proofs. In the present work, we initiate this effort. We define and we prove vital properties of wedges and entanglement wedges, as fully general theorems in the language of Discrete Max-Focusing. Elsewhere, a robust singularity theorem~\cite{Bousso:2025xyc} has already been formulated and proven at the full level of generality and minimality of ingredients advocated here. The Discrete Max formulation and proof of other theorems, such as Gao-Wald or causal wedge inclusion, will be pursued in separate works. 

In Sec.~\ref{sec:nonexpandingwedges}, we re-develop some important concepts and results in semiclassical gravity and holography, using only the criterion of Discrete Max-Nonexpansion instead of a numerical quantum expansion, and using Discrete Max-Focusing in place of the original QFC. In Sec.~\ref{sec:antinormal} we introduce the concepts of an antinormal wedge (one that is both past- and future-nonexpanding), and of a normal wedge (one whose complement is antinormal). We prove that the union of antinormal wedges is antinormal. 

In Sec.~\ref{sec:stablymarginal} we define marginal wedges. Traditionally these correspond to wedges in which at least one quantum expansion vanishes, so that the edge of the wedge lies on an apparent horizon. If both of the quantum expansions vanish, the edge is traditionally called a quantum extremal surface. In our discrete setting, we do not wish to start by imposing numerical conditions like $\Theta^+=0$, since a numerical expansion is not always defined. Nevertheless, an appropriate notion of marginal wedge can be recovered: namely, as an antinormal wedge that admits no future-outward deformation (or no past-outward deformation) that is antinormal. In fact this property implies something more specific and more relevant than just the vanishing of expansions: it defines the notions of stably future- and past-marginal~\cite{Engelhardt:2018kcs}, and of throat (as opposed to bulge or bounce) wedges~\cite{Brown:2019rox,Engelhardt:2023bpv}. We prove that at points where a stably future-marginal wedge has a numerical value of the past max-expansion, it vanishes. Similarly, both numerical expansions of throat wedges vanish where they exist. (Since quantum extremal surfaces obtained by Wall's maximin construction~\cite{Wall:2012uf} are throats~\cite{Brown:2019rox}, this result implies that the standard QES prescription becomes a good approximation to the max- and min- entanglement wedges we define, in those settings where the von Neumann entropy provides a good approximation. This includes in particular the derivation of the Page curve for evaporating black holes from entanglement islands~\cite{Penington:2019npb, Almheiri:2019hni}).

In Sec.~\ref{sec:entanglement}, we define the notion of max- and min entanglement wedges of arbitrary spacetime regions. Again, we use only the notion of Discrete Max-Nonexpansion instead of a numerical quantum expansion; and we use Discrete Max-Focusing in place of the original QFC. We consider the fully dynamical case. Our treatment is completely general, encompassing the ``generalized'' entanglement wedges in arbitrary (non-AdS) spacetimes~\cite{Bousso:2022hlz,Bousso:2023sya}, from which the entanglement wedges of AdS boundary regions emerge as limiting cases~\cite{Bousso:2023sya}.\footnote{Ref.~\cite{Bousso:2023sya} approximated matter entropies as von Neumann entropies; and Ref.~\cite{Akers:2023fqr} only considered entanglement wedges of AdS boundary regions. By considering single-shot entropies and allowing arbitrary spacetimes throughout, we generalize both of these works.} We first reproduce key properties such as entanglement wedge nesting, no-cloning, and the fact that the max entanglement wedge is always contained in the min entanglement wedge. 

Finally we prove two new theorems: when the relevant max and min entanglement wedges of three wedges agree, the entanglement wedges themselves obey strong subadditivity of the \emph{generalized} smooth conditional max and min entropies. These results are highly nontrivial and provide substantial evidence that entanglement wedges represent spacetime regions with definable quantum states. Their proofs require a full arsenal of chain rules, including rules that mix max- and min- entropies~\cite{Vitanov_2013}.

\section{Preliminary Definitions and Conjectures}
\label{sec:prelim}

\subsection{Wedges and Causal Structure}
\label{sec:wedges}

\begin{defn}
  Let $(M,g)$ be a globally hyperbolic manifold with Lorentzian metric. A curve is a continuous map from an interval $C\subset \mathbb{R}$ into $M$. A curve is called timelike (causal) if its tangent vector is everywhere timelike (non-spacelike). 
\end{defn}

\begin{defn}
Let $s\subset M$.\footnote{In this paper, $\subsetneq$ denotes a proper subset; $\subset$ permits equality.} We use $\setint s$, $\cl s$, and $\partial s$ to denote the interior, the closure, and the boundary of $s$ in $M$. 
\end{defn}

\begin{defn}
The chronological future of $s$, $I^+(s)$, is the set of points $q$ such that there exists a future-directed timelike curve that begins at some point $p\in s$ and ends at $q$. $I^-(s)$ is defined analogously. The chronological domain of influence of $s$, $I(s)$, is the set of points that lie on a timelike curve through some point $p\in s$, $I(s)=I^+(s)\cup I^-(s)$.
Similarly, the causal future of $s$, $J^+(s)$, is the set of points that are reached by a future-directed causal curve from a point in $s$. $J^-(s)$ is defined analogously.
\end{defn}

\begin{rem}\label{rem:curves}
    A causal curve need not contain more than one point; hence $J^+(s)\supset s$. A timelike curve contains more than one point, or else its tangent vector would vanish. Hence, $I(s)$ is an open set, $I(s)= I(\cl s)$, and $I(s)$ need not include $s$. For example $I(s)\cap s =\varnothing$ if $s$ is achronal, as would be the case if $s$ is the edge of a wedge, defined below. 
\end{rem}
\begin{defn}\label{def:sc}
The \emph{spacelike complement} of a set $s\subset M$ is defined by
\begin{equation}\label{eq:sc}
    s'\equiv \setint [M\setminus I(s)]~.
\end{equation}
(Thus, $s'$ is necessarily open.)
\end{defn}

\begin{defn}\label{def:covwedge}
A {\em wedge} is a set $a\subset M$ that satisfies $a=a''$.  (This immediately implies that  $a$ is an open set; $a'$ is also a wedge; and the intersection of two wedges $a,b$ is a wedge~\cite{Bousso:2022hlz, Bousso:2023sya}.)
\end{defn}

\begin{defn}\label{def:wedgeunion}
The {\em wedge union} of two wedges $a,b$ is the wedge
\begin{equation}
    a\Cup b\equiv (a'\cap b')'~.
\end{equation}
\end{defn}

\begin{defn}\label{def:edgehor}
The \emph{edge} $\eth a$ and \emph{Cauchy horizons} $H^\pm(a)$ of a wedge $a$ are
defined by 
\begin{align}
    \eth a & \equiv \partial a \setminus I(a)~,\\
    H^+(a) & \equiv \partial a\cap I^+(a)~,\\
    H^-(a) & \equiv \partial a\cap I^-(a)~,\\
    H(a) & \equiv \partial a\cap I(a) = H^+(a)\cup H^-(a)~.
\end{align}
\end{defn}

\begin{rem}
Any wedge $a$ induces a decomposition of the spacetime $M$ into disjoint sets:
\begin{align}\label{eq:decom}
    M & = a \sqcup a' \sqcup \eth a \sqcup I(\eth a) \sqcup H(a) \sqcup H(a') 
    ~.
\end{align}
\end{rem}

\begin{rem}\label{rem-ethprops}
    $\eth a$ is closed in $M$ since $\eth a=M\setminus I(a)\setminus I(a')$. Moreover, $\eth a$ is an immersed $C^0$ submanifold of codimension 2. However, $\eth a$ need not be embedded (for example, in $2$ spatial dimensions, suppose that $a$ is the union of two disjoint wedges whose edges intersect at a point).
\end{rem}
\begin{defn}
    The area of $\eth a$ will be denoted $\A(a)$.
\end{defn}
\begin{defn}\label{def:conformaledge}
    Given a wedge $a$, we distinguish between its edge $\eth a$ in $M$ and its edge $\deltabar a$ as a subset of the conformal completion $\tilde M$~\cite{Wald:1984rg}. If $a$ is asymptotic, the latter set may contain an additional piece, the {\em conformal edge}
    \begin{equation}
        \tilde\eth a\equiv \deltabar a\cap \partial\tilde M~.
    \end{equation}
\end{defn}

\subsection{Generalized Conditional Max Entropy and Its Properties}
\label{sec:hmaxgen}

\begin{conj}\label{conj:hmaxgenexists}
    Let $M$ be a spacetime that satisfies Einstein's equation $G_{ab}=8\pi G\braket{T_{ab}}$ to all orders in $G\hbar$, where $G$ is Newton's constant. For any nested pair of wedges $a\subset b\subset M$, there exists a family, parametrized by $\epsilon$, of \emph{generalized smooth conditional max entropies}, $\hmg^\epsilon(b|a)$, that satisfy
    \begin{equation}\label{eq:hmgexpansion}
        \hmg^\epsilon(b|a) = \left[ \frac{\A(b)-\A(a)}{4G} + \hmax^\epsilon(b|a) + O(G)\right]~,
    \end{equation}
    where $\hmax^\epsilon(b|a)$ is the standard smooth max entropy of the matter field quantum state restricted to $b$ conditioned on $a$~\cite{RenWol04a,Akers:2020pmf,Akers:2023fqr}. 
\end{conj}    

\begin{rem}
Regulating each term in Eq.~\eqref{eq:hmgexpansion} requires a short-distance cutoff on which both Newton's constant and $\hmax(b|a)$ will depend. Based on suggestive results for the generalized von Neumann entropy (see the Appendix of Ref.~\cite{Bousso:2015mna}), we conjecture that the sum, $\hmg^\epsilon(b|a)$, is independent of the cutoff. 
\end{rem}

\begin{rem}
    A more careful algebraic definition of $\hmg$ allows for fluctuations of the area by treating it as an operator~\cite{Akers:2023fqr}. Then there are quantum states for which the two terms in Eq.~\eqref{eq:hmgexpansion} do not separate neatly. This is irrelevant for the purposes of our work. The only properties of $\hmg$ that we will need are those listed in the remainder of this section.
\end{rem}

\begin{defn}
    The family of generalized smooth min-entropies is defined by
    \begin{equation}\label{eq:minMaxDef}
        \hmingen^\epsilon(b|a) = -\hmg^\epsilon(a'|b')~.
    \end{equation}
\end{defn}

\begin{rem}
    Since the ordinary smooth conditional min entropy satisfies $\hmin^\epsilon(b|a) = -\hmax^\epsilon(a'|b')$, and since $\A(a)=\A(a')$, $\A(b)=\A(b')$, it follows that
    \begin{equation}
        \hmingen^\epsilon(b|a) = \left[ \frac{\A(b)-\A(a)}{4G} + \hmin^\epsilon(b|a) + O(G)\right]~.
    \end{equation}
\end{rem}

\begin{conv}
From here on we shall suppress the superscript $\epsilon$ on $\hmg$, $\hmingen$, $\hmax$, and $\hmin$. Equations below omit correction terms that depend on $\epsilon$. See e.g.\ Ref.~\cite{Vitanov_2013, Tomamichel:2012xyw} for a statement of the chain rules that includes these terms.
\end{conv}

\begin{rem}
    We will now conjecture that certain key properties of the standard smooth conditional entropies hold also for the generalized versions. The following conjectures are manifestly true for the two leading terms displayed in the above equations. They are conjectural only in that they insist that the higher-order terms do not spoil each property. 
\end{rem}

\begin{conj}[Generalized asymptotic equipartition principle]\label{conj:aep}
    Let $M$ be a spacetime satisfying the Einstein equation as in Conjecture~\ref{conj:hmaxgenexists} and let $a\subset b$ be wedges. Then we can consider the generalized max and min entropies of $k$ copies of $b$ conditioned on $k$ copies of $a$, in the $k$-fold replicated spacetime $M^{\otimes k}$ (with $k$-fold replicated quantum state of the matter fields). At the level of the leading terms displayed in Eqs.~\eqref{eq:hmgexpansion} and \eqref{eq:minMaxDef}, the quantum equipartition principle is guaranteed to hold. We conjecture that it holds for the full quantities:
    \begin{equation}
        \lim_{k\to\infty}\frac{1}{k}\hmingen(b^{\otimes k}|a^{\otimes k})_{M^{\otimes k}} = \S(b|a)_M = \lim_{k\to\infty}\frac{1}{k}\hmg(b^{\otimes k}|a^{\otimes k})_{M^{\otimes k}}~.
    \end{equation}
    That is, for large $k$, the max and min conditional generalized entropies in the replicated setup approach $k$ times the von Neumann conditional generalized entropy of a single copy.
\end{conj}

\begin{conj}\label{conj:minLeqMax}
    For any wedges $a\subset b$,
    \begin{equation}
    \hmingen(b|a) \leq \S(b|a) \leq \hmg(b|a)~.
    \end{equation}
    This holds trivially for the area difference, and it holds as a standard property of the quantum conditional entropies~\cite{Vitanov_2013}. We conjecture that it holds for the full quantities.
\end{conj}

\begin{conj}[Chain rule]\label{conj:chain}
    Let $a\supset b\supset c$. The generalized max and min entropies obey the following inequalities:
    \begin{align}
        \hmg(a|c) &\leq\hmg(a|b)+\hmg(b|c) \label{eq:chainMaxMaxMax}\\
        \hmingen(a|c) &\geq \hmingen(a|b)+\hmingen(b|c) \label{eq:chainMinMinMin}\\
        \hmingen(a|c) &\leq \hmg(a|b) + \hmingen(b|c)\label{eq:chainMinMaxMin}\\
        \hmg(a|c) &\geq \hmg(a|b) + \hmingen(b|c)\label{eq:chainMaxMaxMin}~.
    \end{align}
    These relations again hold trivially for the area terms; and they hold by Ref.~\cite{Vitanov_2013} for the conditional max and min matter entropies. We conjecture that they hold for the full quantities.
\end{conj}

\begin{conj}[Strong Subadditivity of the Generalized Conditional max Entropy]\label{conj:ssa}
    Let $a\subset b$ be wedges. Then for any wedge $c$ spacelike to $b$,
    \begin{equation}
        \hmg(b\Cup c|a\Cup c)\leq \hmg(b|a)~.
    \end{equation}
\end{conj}

\begin{rem}\label{rem:whyssa}
    Like all the above conjectures, this holds trivially for the area terms, and as a mathematical theorem for the conditional max entropy of the quantum fields~\cite{Bousso:2023sya}:
    \begin{align}
    \A(b\Cup c)-\A(a\Cup c) & \leq \A(b)-\A(a) ~;\\
    \hmax(b\Cup c|a\Cup c) & \leq \hmax(b|a)~.
    \end{align}
    We conjecture that it holds for the full quantity $\hmg$.
\end{rem}

\begin{rem} \label{rem:ssamin}
    Strong subadditivity of the generalized conditional min entropy follows directly from Conjecture~\ref{conj:ssa} and Eq.~\eqref{eq:minMaxDef}.
\end{rem}

\begin{conj}[Discrete Subadditivity of the Generalized Conditional max entropy]\label{conj:dsa}
    Let $a\subset b$ be wedges.
    \begin{enumerate}[i.]
        \item If $\hmg(b|a)\leq 0$, then for any wedge $c$ spacelike to $b$, $\hmg(b\Cup c|a\Cup c)\leq 0$.
        \item If $\hmg(b|a)<0$, then for any wedge $c$ spacelike to $b$, $\hmg(b\Cup c|a\Cup c)<0$.
    \end{enumerate}
\end{conj}

\begin{rem}
    Conjecture~\ref{conj:dsa} follows from Conjecture~\ref{conj:ssa}. Based on Remark~\ref{rem:whyssa}, we find the stronger Conjecture~\ref{conj:ssa} quite plausible. Nevertheless we include the weaker Conjecture~\ref{conj:dsa}, because it suffices for almost all of our results. Only Theorem~\ref{thm:ssa} requires Conjecture~\ref{conj:ssa}. The somewhat weaker Corollary~\ref{cor:dsa} can still be obtained using only Conjecture~\ref{conj:dsa}.
\end{rem}

\section{Discrete Max-Nonexpansion and Discrete Max-Focusing}
\label{sec:nonexp}

In this section, we introduce a discrete definition of ``non-expanding'' that is more robust than the traditional method of defining a numerical (quantum-)expansion and then demanding that it be nonpositive. We use the conditional max entropy throughout as this will be the only kind of non-expansion needed in the following sections. We will omit the words ``discrete'' and ``max'' in most places since no other expansions are defined.

\subsection{Definitions}
\label{sec:discrete}

\begin{defn}[Nonexpansion]\label{def:fne}
A wedge $a$ is said to be \emph{future-nonexpanding} at $p\in\eth a$ if there exists an open set $O$ containing $p$ such that 
    \begin{align}\label{eq:futnonexp}
        \hmg(b|a) \leq 0 \text{~for~all~wedges~}b\supset a\text{~such~that~} 
        \eth b \subset \eth a \cup [H^+(a')\cap O]~.
    \end{align}
We emphasize that no numerical value is assigned to the expansion by the present definition; see Def.~\ref{def:thmax} below. To define \emph{past-nonexpanding}, $H^-$ replaces $H^+$ in Eq.~\eqref{eq:futnonexp}.
\end{defn}

\begin{defn}[Noncontraction]
    A wedge $a$ is \emph{future-noncontracting} if there exists an open set $O\supset\eth a$ such that no proper past-directed outward null deformation of $a$, with compact support within $O$, is future-nonexpanding on all new edge points. That is, no wedge $b\supsetneq a$ with $\eth b\subset \eth a \cup [H^-(a')\cap O]$ is future-nonexpanding on all points in $\eth b\setminus\eth a$. (Heuristically, in sufficiently smooth settings, future-noncontracting implies nonnegativity of the outward future quantum expansion of $a$; moreover, positivity of the quantum expansion implies future-noncontracting.) \emph{Past-noncontracting} is defined analogously.
\end{defn}

\begin{defn}
    Let $\eth a^+$ ($\eth a^-$) denote the set of points $p\in \eth a$ where the wedge $a$ is future- (past-) nonexpanding.
\end{defn}

\begin{lem}\label{lem:ethaplusopen}
    $\eth a^+$ and $\eth a^-$ are open subsets of $\eth a$ in the induced topology of $\eth a$.
\end{lem}

\begin{proof}
    Let $p\in \eth a^+$. We must show that $\eth a^+$ contains an open subset that contains $p$. Let $O$ be the open set guaranteed to exist by Def.~\ref{def:fne}, and let $q\in \eth a\cap O$. Then $O$ contains $q$ and satisfies Eq.~\eqref{eq:futnonexp}. Hence all points in $\eth a\cap O$ are future-nonexpanding: $\eth a\cap O\subset \eth a^+$.
\end{proof}

\begin{defn}[Lightsheet]
    Let $a$ be a wedge and let $\eth a^+$ be the set of points where $a$ is future-nonexpanding. The future-lightsheet of $a$ is the null hypersurface
    \begin{equation}
        L^+(a)\equiv H^+(a') \cap J^+(\eth a^+)~.
    \end{equation}
    The past light-sheet is defined analogously.
\end{defn}

\vspace{5mm}

\hspace{-8.5mm} 
\begin{minipage}{0.45\textwidth}
    \includegraphics[width=1.2\textwidth]{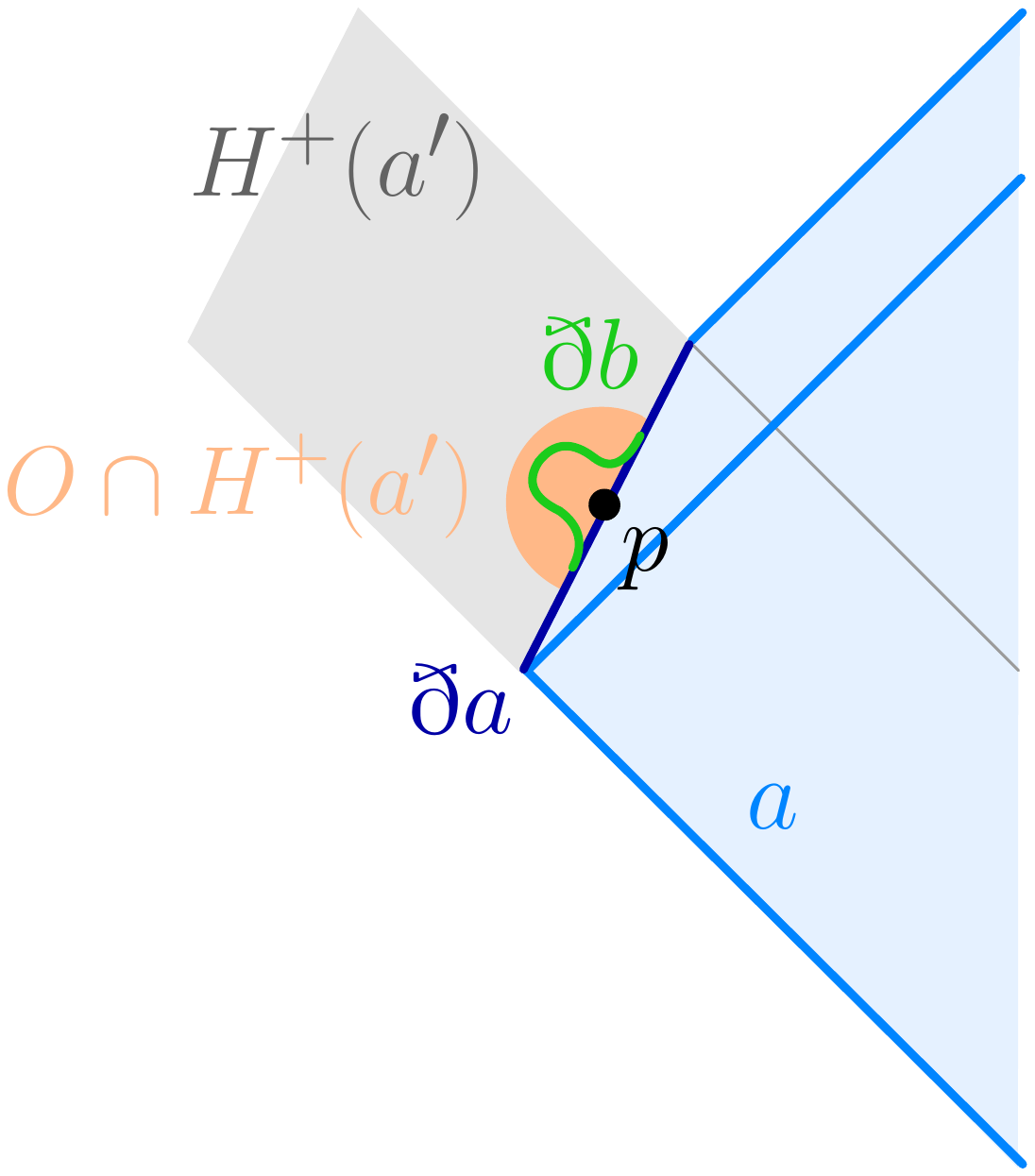}
    \vspace{3mm}
    \captionof{figure}{Illustration of Definition~\ref{def:fne}: a wedge $a$ is future-nonexpanding at $p$ if there exists an open neighborhood $O$ of $p$ such that all future-outward deformations of $\eth a$ within $O$ (green) do not increase the area (classically), or (semiclassically) have nonpositive smooth max-generalized entropy conditioned on $a$.}
    \label{fig:fne}
\end{minipage}
\hspace{9mm}
\begin{minipage}{0.48\textwidth}
    \includegraphics[width=1.5\textwidth]{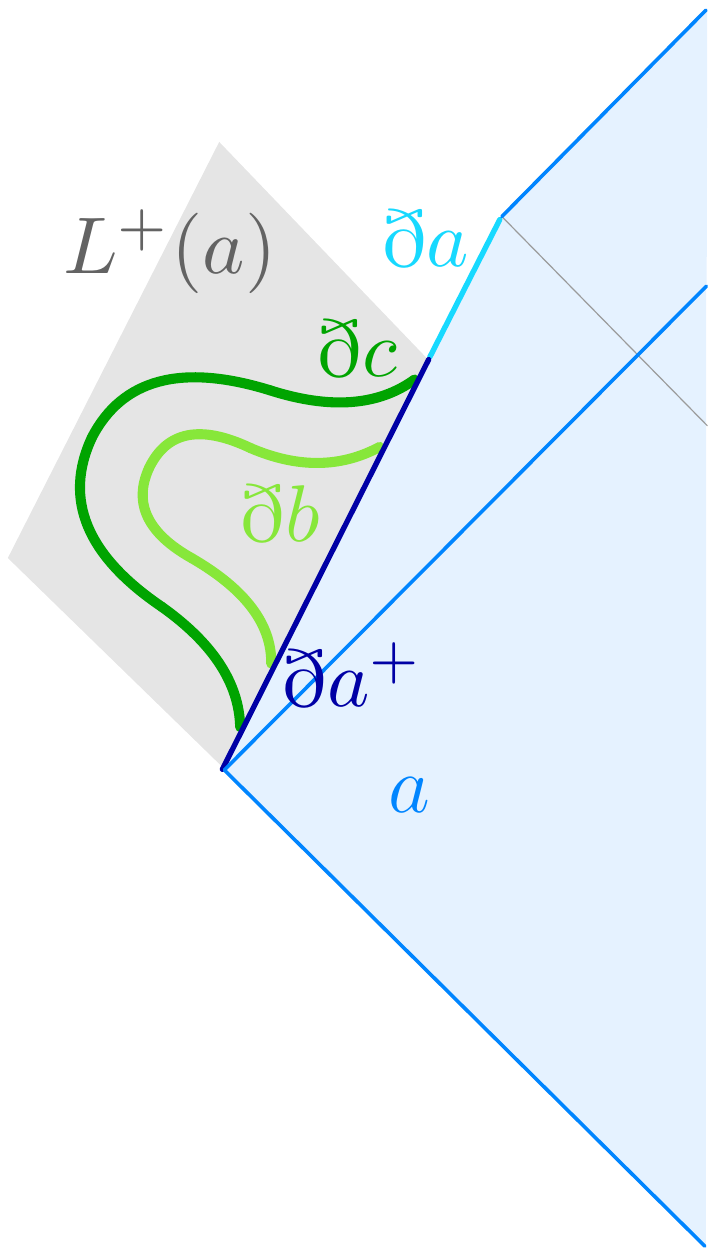}
    \captionof{figure}{Discrete Max-Focusing. If the wedge $a$ is future-nonexpanding on a portion $\eth a^+$ of its wedge $\eth a$, then $a$ can be outward deformed along the future lightsheet emanating from $\eth a^+$ into wedges $b$, $c$ that obey $\hmg(c|b)\leq 0~.$}
\end{minipage}
\vspace{3mm}

\begin{conj}[Discrete Max-Focusing]\label{conj:qfc}
    Let $a$, $b$, and $c$ be wedges such that $a\subset b\subset c$, and suppose that $\eth b \cup \eth c \subset \eth a \cup L^+(a)$. Then 
    \begin{equation}\label{eq:qfc}
        \hmg(c|b)\leq 0~.
    \end{equation}
    The same statement holds in the past direction, i.e., if $\eth b \cup \eth c \subset \eth a \cup L^-(a)$.
\end{conj}

\subsection{Immediate Implications}\label{sec:imply}

\begin{lem}
    A wedge $a$ is future- and past-nonexpanding at any point $p$ in the interior of a null portion of its edge.
\end{lem}

\begin{proof}
    By assumption there exists an open neighborhood $O$ of $p$ in which the past and future null vectors normal to $\eth a$ are tangent to $\eth a$. None of the null generators of $H^\pm(a')$ start on this portion. Since $a$ cannot be deformed along $H^\pm(a')$ at all, there exists no wedge $b$ for which Eq.~\eqref{eq:futnonexp} could fail.
\end{proof}

\begin{figure}
    \centering
    \hspace{-50mm}\includegraphics[width=7.5cm]{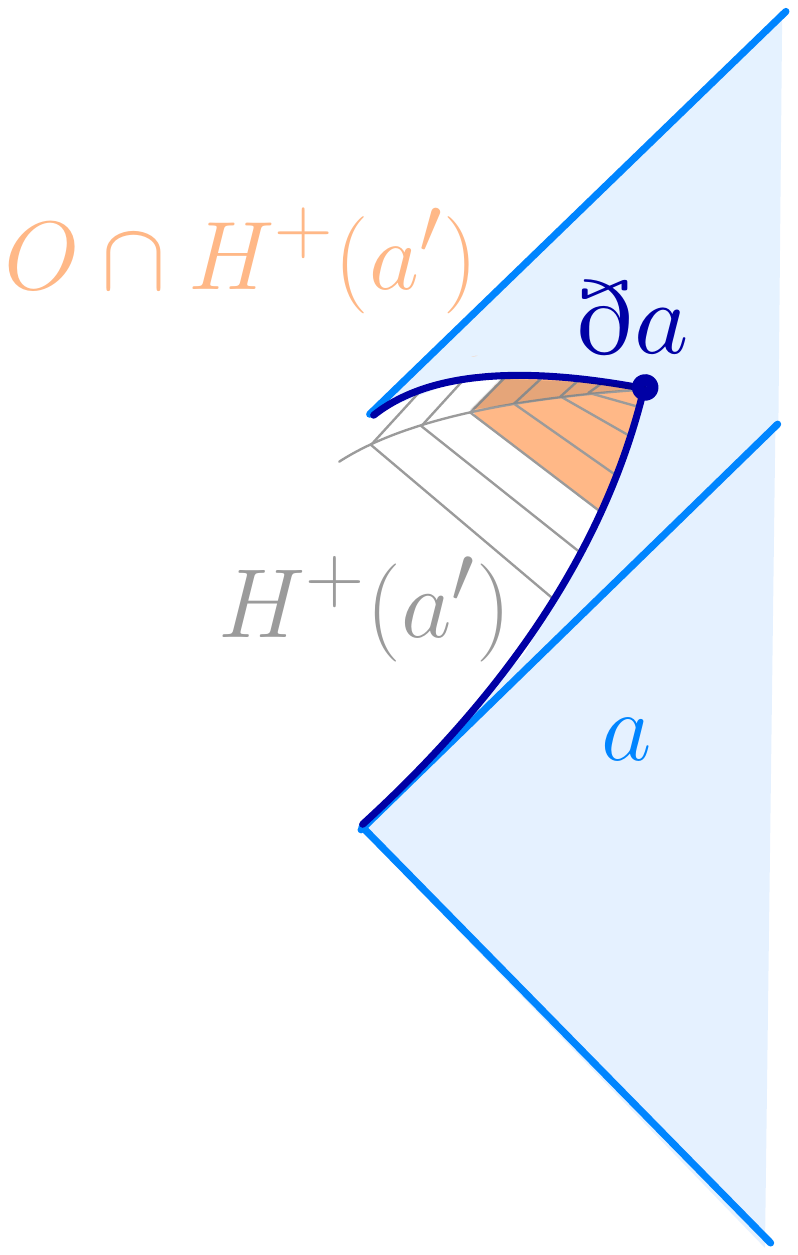}\hspace{-5mm}
    \includegraphics[width=7.5cm]{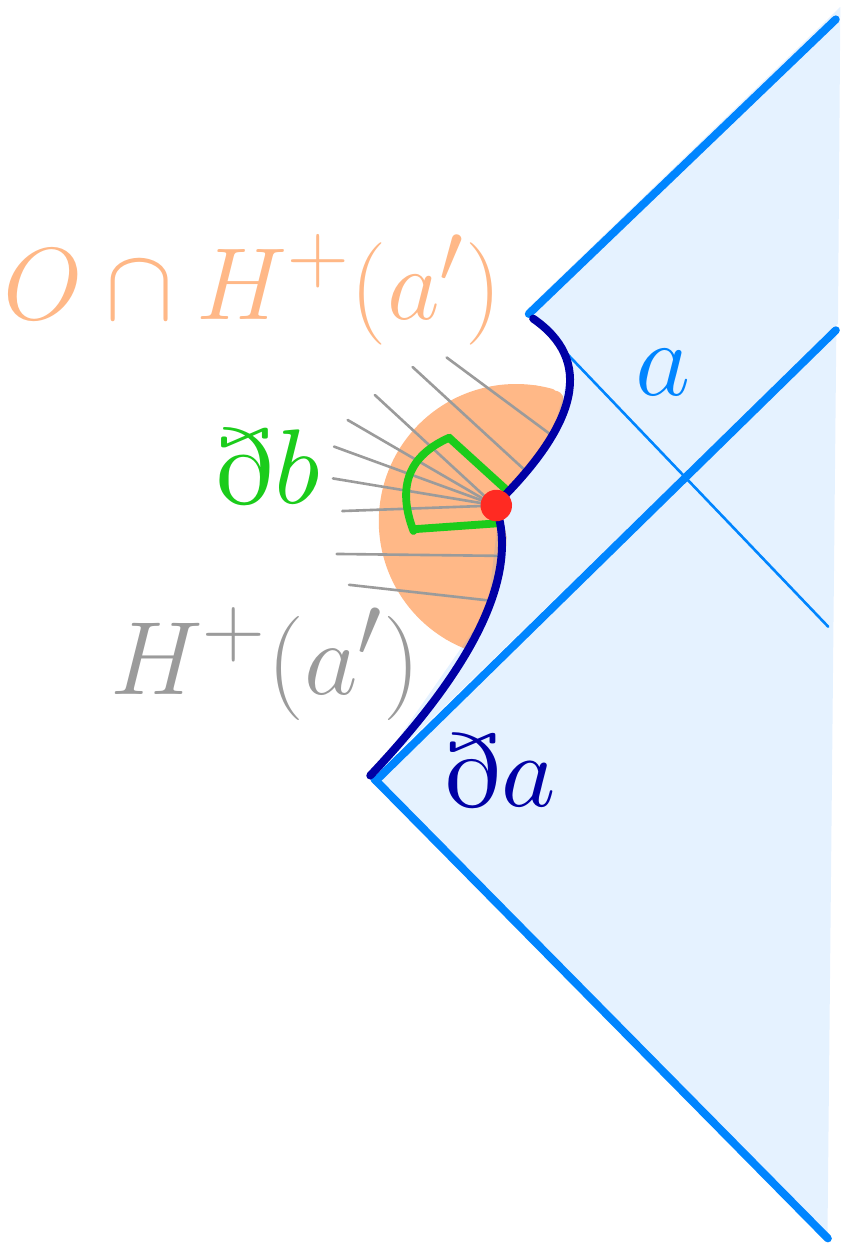}\hspace{-25mm}
    \includegraphics[width=7.5cm]{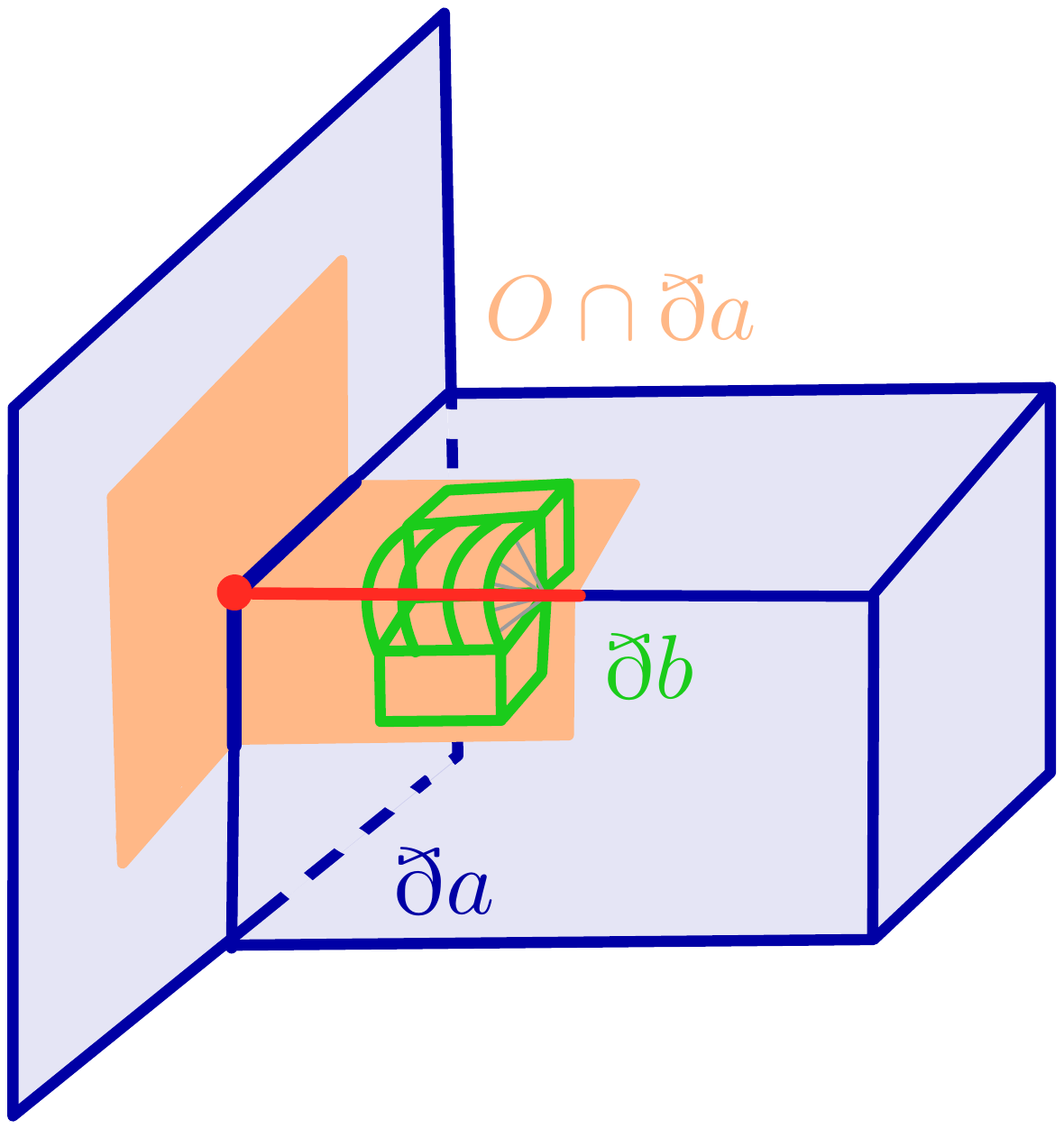}\hspace{-30mm}
    \caption{Applications of Definition~\ref{def:fne}. All wedges shown are understood to be subsets of Minkowski spacetime. \emph{Left:} the wedge $a$ is nonexpanding at all the dark blue points, including at the (concave) corner (blue dot). \emph{Center:} $a$ is future-nonexpanding except at the (convex) corner, where the green deformation $\eth a\to\eth b$ fails to satisfy Eq.~\eqref{eq:futnonexp}. \emph{Right:} the time direction is suppressed, but all three spatial directions are now displayed. The edge $\eth a$ is now shown as a shaded surface. The corner (red dot) is neither convex nor concave. Such a corner fails to be future-nonexpanding, because in any open neighborhood (orange), the green deformation increases the area.}
    \label{fig:fneexamples}
\end{figure}

\begin{thm}[Persistence of Nonexpansion]\label{thm:persistence}
    Let $a$ and $b$ be wedges such that $a\subset b$, and suppose that $\eth b \subset \eth a \cup L^+(a)$. Assuming Discrete Max-Focusing, Conjecture~\ref{conj:qfc}, 
    \begin{equation}
        \eth b \cap L^+(a)\subset \eth b^+~.
    \end{equation}
    The same statement holds in the past direction: $+$ can be replaced by $-$ throughout.
\end{thm}

\begin{proof}
    Let $p\in \eth b \cap L^+(a)$. To meet Def.~\ref{def:fne}, we must show that there exists an open set $O$ containing $p$ such that
    \begin{equation}\label{eq:cbproof}
        \hmg(c|b) \leq 0 \text{~for~all~wedges~}c\supset b\text{~such~that~} 
        \eth c \subset \eth b \cup [H^+(b')\cap O]~.
    \end{equation}

    By Lemma~\ref{lem:ethaplusopen}, $\eth a\setminus \eth a^+$ is closed in the induced topology of $\eth a$. Hence $\eth b\cap L^+(a) = \eth b \setminus (\eth a\setminus \eth a^+)$ is open in the induced topology of $\eth b$. Therefore, there exists a set $O$ that is open in $M$ such that $\eth b \cap O = \eth b\cap L^+(a)$.

    Let $c\supset b$ such that $\eth c \subset \eth b\cup [H^+(b')\cap O]$. Since $H^+(b')\subset H^+(a')$, we have $\eth c \subset \eth a \cup L^+(a)$ and hence the assumptions of Discrete Max-Focusing, Conjecture~\ref{conj:qfc}, are satisfied. The Conjecture then implies Eq.~\eqref{eq:cbproof}.
\end{proof}

\begin{thm}[Bousso Bound]
    Let $a\subset c$ such that $\eth c\subset\eth a\cup L^+(a)$. That is, $c$ is a future-outward deformation of $a$ along orthogonal lightrays originating on the future-nonexpanding portion of its edge $\eth a$. Then
    \begin{equation}
        \hmg(c|a)\leq 0~.
    \end{equation}
    The same statement also holds if $\eth c\subset\eth a\cup L^-(a)$.
\end{thm}
\begin{proof}
    The result is a special case of Discrete Max-Focusing, Conjecture~\ref{conj:qfc}, obtained by setting $b=a$ in Eq.~\eqref{eq:qfc}.
\end{proof}

\begin{defn}[Future-Causal Horizon]
    Let ${\cal C}$ be the boundary of the past of a subset $s$ of the conformal boundary of an arbitrary spacetime $M$. (Note that $s$ can be, but need not be, a single point.) If there exists a Cauchy slice $\Sigma$ of $M$ such that any wedge $a$ with $\eth a\subset {\cal C}\cap J^+(\Sigma)$ is past-nonexpanding, then $\cal C$ is called a \emph{future-causal horizon}.
\end{defn}
\begin{rem}
    Let ${\cal C}$ be the boundary of the past of a subset of conformal infinity of an asymptotically flat spacetime $M$. We expect that the above property will then hold, so that ${\cal C}$ is a future-causal horizon. We use the above definition because it is more general and because it suffices for excluding ``horizons'' defined in terms of points on a singularity such as that inside a Schwarzschild black hole. 
\end{rem}
\begin{thm}[Generalized Second Law]
    Let $\cal C$ be a future-causal horizon. Let $b\subset c$ be wedges whose edges both lie on $\cal C$ with $\eth c\subset L^-(b)$. Then
    \begin{equation}
        \hmg(c|b)\leq 0~.
    \end{equation}
\end{thm}
\begin{proof}
    By the assumption, there exists a past-nonexpanding wedge $a$ whose edge lies on $\cal C$ and to the future of $b$. The result follows immediately from Discrete Max-Focusing, Conjecture~\ref{conj:qfc}.
\end{proof}

\subsection{Special Case: Compressible States}
One often considers ``compressible'' quantum states, for which
\begin{equation}
    \hmg(b|a) = \S(b|a) = \hmingen(b|a)~,
\end{equation}
where $\S(b|a)$ is the conditional generalized von Neumann entropy. Unlike the generalized max and min entropies, this can be written as a difference of generalized entropies:
\begin{equation}
    \S(b|a) = \S(b)-\S(a)~.
\end{equation}
We stress that in such settings, our novel definition of \emph{future-nonexpanding}, Def.~\ref{def:fne}, and our Discrete Max-Focusing Conjecture, Conj.~\ref{conj:qfc}, are both still needed. Their purpose is to deal with edge points where the \emph{classical} expansion is not well-defined, and this problem is unrelated to the quantum corrections to the area term.

\subsection{Special Case: Numerical Max-Expansion}
\label{sec:continuous}

So far, we have only defined qualitative notions of past and future nonexpansion, which does not assign a numerical value to the expansion of light-rays orthogonal to the edge of a wedge. As we explained in the introduction, numerical values cannot be assigned at null portions on the edge of a wedge, and at points where the edge is insufficiently smooth. Such points arise generically from the caustic structure of null hypersurfaces, since wedges are often defined by cuts of such hypersurfaces. They arise also when we consider unions or intersections of wedges. In this subsection, we will consider a more restrictive class of wedges for which a numerical max-expansion can be defined, and we relate its value to the more general qualitative notion of nonexpansion.

\begin{defn}[Numerical Expansion]\label{def:thmax}
    Let $a$ be a wedge whose edge $\eth a$ is acausal and $C^2$-smooth. (That is, no two points in $\eth a$ are causally connected; and $\eth a$ is a $C^2$ embedded submanifold.) Let $y$ collectively denote the transverse coordinates that label points in $\eth a$, and let $h$ be the determinant of the induced metric on $\eth a$. Let $V:\eth a\to [0,1]$, $y\to V(y)$ be a piecewise continuous function. Let $a^+_{V,\lambda}$ be the outward deformation of $a$ along the unique null generator of $H^+(a')$ that originates at each point $y$ in $O$, by an affine distance $\lambda V(y)$. (Note that $a^+_{V,\lambda}$ exists for all sufficiently small $\lambda$.) Then there exists a continuous function $\thmax^+:O\to \mathbb{R}$ such that
    \begin{equation}\label{eq:thmaxdef}
        \lim_{\lambda\to 0} \frac{1}{\lambda}\hmg(a^+_{V,\lambda}|a) = 
        4G\int_O d^{d-2} y\sqrt{h} \, V(y) \thmax^+(y)~,
    \end{equation}
    for every choice of $V$. We call $\thmax^+(a;p)\equiv \thmax^+(y)$ the \emph{future expansion} of $a$ at the point $p\in\eth a$ labeled by the coordinates $y$. The \emph{past expansion} is defined analogously.
\end{defn}

\begin{rem}
    The above definition can be generalized somewhat, though we will not try to do this carefully here: if $\eth a$ is not acausal or not smooth (or neither), $\Theta^\pm(a,p)$ can still be defined at any point $p\in \eth a$ that is acausal to the rest of $\eth a$ and possesses an open neighborhood $O$ such that $\eth a\cap O$ is a smooth embedded submanifold. Moreover, even at a point $p$ where $\eth a$ is not $C^2$, one of the expansions may exist: namely, if one of $H^\pm(a')$ is smooth in a neighborhood of $p$. 
\end{rem}

\begin{thm}\label{thm:fnetononpos}
    Suppose that the wedge $a$ is future-nonexpanding at $p$ and that $\thmax^+(a;p)$ exists. Then $\thmax^+(a;p)\leq 0$.
\end{thm}
\begin{proof}
    Since $a$ is future-nonexpanding at $p$, there exists an open neighborhood $O$ of $p$ such that any future-outward deformation $b$ along $H^+(a)$ that remains within $O$ obeys $\hmg(b|a)\leq 0$. Let $V$ and $a^+_{V,\lambda}$ be defined as in Def.~\ref{def:thmax}. Since $\hmg(a^+_{V,\lambda}|a)$ is nonpositive for sufficiently small $\lambda$, and since Eq.~\eqref{eq:thmaxdef} holds for all nonpositive functions $V:O\to \mathbb{R}_0^+$, one finds that $\thmax^+\leq 0$.
\end{proof}

\begin{thm}\label{thm:negtofne}
    Suppose that $\thmax^+(a;p)$ exists and that $\thmax^+(a;p)<0$. Then $a$ is future-nonexpanding at $p$.
\end{thm}

\begin{proof}
    By Def.~\ref{def:thmax}, there exists an open neighborhood $O$ of $p$ such that $\thmax^+(a;q)<0$ for all $q\in O$. We may now choose any piecewise continuous function $V:0\to [0,1]$, giving rise to a family of deformations $a^+_{V,\lambda}$ as described in Def.~\ref{def:thmax}. Since $\thmax^+(a;p)$ is strictly negative, by continuity, $\hmg(a^+_{V,\lambda}|a)$ must be negative for all sufficiently small $\lambda$. After restricting $O$ to such values for all possible choices of $V$, Eq.~\eqref{eq:thmaxdef} implies that Def.~\ref{def:fne} is met.
\end{proof}

\begin{rem}
    Note that $\thmax^+(a;p)= 0$ does not imply that $a$ is future-nonexpanding at $p$. For example, $\thmax^+(a;p)= 0$ does not preclude that $\hmg(a^+_{V,\lambda}|a)>0$ for all $\lambda>0$. 
\end{rem}

\section{Evidence for Discrete Max-Focusing}
\label{sec:evidence}

\subsection{No Geodesics Enter Lightsheets}
\label{sec:noentry}

Discrete Max-Focusing would be violated already at the level of differential geometry if it was possible for new null geodesics to enter $L^+(a)$, i.e., if we could choose $\eth c$ to intersect null geodesics that were not present on $\eth b$. Then by choosing $\eth c$ to include only a small neighborhood of these geodesics and otherwise agree with $\eth b$, the area contribution to $\hmg(c|b)$ would be expected to dominate and would be positive. 

However, this is excluded by Theorem 9.3.11 of Ref.~\cite{Wald:1984rg}; see also~\cite{Akers:2017nrr}. Note in particular that this excludes the development of ``corners'' or folds on which null geodesics would enter $H^+(a')$. For example, spatially concave corners can (and generically do) develop along $H^+(a')\cap \eth a^+$, but spatially convex corners cannot, since null geodesics would enter there.

\subsection{Classical Focusing}
\label{sec:classical}

Let us define classical General Relativity by taking $\hbar\to 0$ and requiring that the metric $g_{\mu\nu}$ of the manifold $M$ satisfies the Einstein equation with the Null Energy Condition holding for matter. That is, the stress tensor satisfies $T_{\mu\nu} k^\mu k^\nu \geq 0$.

In the $\hbar\to 0$ limit, only the area terms survive in the smoothed conditional generalized entropy: 
\begin{equation}
    4 G\hbar\, \hmg(b|a) \to \A(b)-\A(a)~.
\end{equation}
By Def.~\ref{def:fne}, $a$ is future-nonexpanding at any point $p$ where $\eth a$ is differentiable, if and only if the expansion $\theta^+$ of the null generators of $H^+(a')$ is nonpositive at $a$. Then both Theorem~\ref{thm:persistence} and Conjecture~\ref{conj:qfc} follow immediately from the Raychaudhuri equation~\cite{Wald:1984rg},
\begin{equation}\label{eq:raych}
    \theta = -\frac{1}{d-2} \theta^2 - \varsigma_{\mu\nu}\varsigma^{\mu\nu} - 8\pi G\, \braket{T_{\mu\nu}}k^\mu k^\nu~,
\end{equation}
where $\varsigma$ is the shear of the null congruence. 

Thus, in the $\hbar\to 0$ limit, Discrete Max-Focusing and the Persistence of Nonexpansion become special cases of the classical focusing theorem~\cite{Wald:1984rg}.

\subsection{Discrete Max-Focusing Implies the QNEC}
\label{sec:qnec}

In a different limit, where $G\to 0$ at finite $\hbar$, we will now argue that Discrete Max-Focusing implies the Quantum Null Energy Condition. The QNEC is a highly nontrivial property of relativistic unitary quantum field theories, originally derived from the von-Neumann QFC. The QNEC was later proven explicitly using advanced methods in quantum field theory~\cite{Bousso:2015wca,Balakrishnan:2017bjg,Ceyhan:2018zfg}. Since we no longer wish to rely on the von Neumann QFC, it is vital that we show that Discrete Max-Focusing, too, implies the QNEC. 

Let $a$ be a wedge in Minkowski space (with $G=0$) whose edge $\eth a$ is acausal and smooth. Let $p\in \eth a$, and let the vector $k^\mu$ be null, future and outward-directed, and orthogonal to $\eth a$ at $p$. We require conditions on the extrinsic curvature of $\eth a$ at $p$ that are necessary for the QNEC to be regularization-scheme independent~\cite{Akers:2017ttv}. In particular, the classical null expansion and shear must vanish at $p$: $\theta^+=\varsigma^+_{\mu\nu}=0$. 
Let $\psi$ be the state of the quantum fields on $M$ (pure or mixed). Now consider $M$ with $n$ copies of the same quantum field theory, in the state $\psi^{\otimes n}$, and turn on $G\neq 0$. Our argument will involve taking the limit as $G$ is turned back off. Therefore, we can take $G$ to be arbitrarily small at the present step.

For sufficiently small $G$ at fixed $n$, a wedge $a(G,n)$ can be identified simply by assigning to it the same coordinates as those of $a$ in the nongravitating spacetime. In the same way, one can identify the point $p(G,n)\in \eth a(G,n)$. Now deform $a(G,n)$ smoothly and acausally in a neighborhood of $p(G,n)$, but not at $p(G,n)$ itself, to a wedge $b(G,n)$ (whose edge is also acausal and smooth, so that $\thmax^+[b(G,n)]$ exists), so as to achieve $\Theta^+_{\rm max}[b(G,n);p(G,n)]=-O(G^2)<0$.

Neither the identifications of $a(G,n)$ and $p(G,n)$, nor the deformation into $b(G,n)$ are unique, even while satisfying the criteria of Ref.~\cite{Akers:2017ttv}. The only important requirement is that they be chosen such that $p(G,n)\to p$, $a(G,n)\to a$, and $b(G,n)\to a$ as $G\to 0$ at fixed $n$.\footnote{To define an appropriate notion of continuity, one can represent the different spacetimes $M(G,n)$ as Minkowski space with a perturbative gravitational field.} The latter property can be satisfied because $\thmax^+[b(G,n);p(G,n)]-\theta^+[a(G,n);p(G,n)]\sim O(G)$.

\vspace{3mm}
\hspace{-8.5mm}
\begin{minipage}{0.55\textwidth}
    \quad~ We may now apply Discrete Max-Focusing to $b(G,n)$. Since $\thmax^+[b(G,n);p(G,n)]<0$ and $b(G,n)$ can be chosen smooth, there exists an open neighborhood $O$ of $p(G,n)$ such that $\thmax^+[b(G,n);q]<0$ for all $q\in \eth b(G,n)\cap O$; hence $\eth b(G,n)\cap O\subset \eth b(G,n)^+$. Let $c(G,n)\supset b(G,n)$ be an outward deformation of $b(G,n)$ along $H^+[b'(G,n)]$ such that $\eth b(G,n)\setminus O= \eth c(G,n)\setminus O$. Let $q\in \eth c$ lie on the future null generator originating at $p$. We may choose $c(G,n)$ such that in the limit as $G\to 0$ at fixed $n$, it approaches a fixed wedge $c$ whose edge lies on $H^+(a')$. (Recall that $a(G,n)\to a$, and $b(G,n)\to a$ as $G\to 0$ at fixed $n$.) Theorem~\ref{thm:persistence} implies
\begin{equation}\label{eq:thmc}
    \thmax^+[c(G,n);q]\leq 0~.
\end{equation}
    
\end{minipage}
\hspace{7mm}
\begin{minipage}{0.38\textwidth}
    \hspace{-8mm}
    \includegraphics[width=1.5\textwidth]{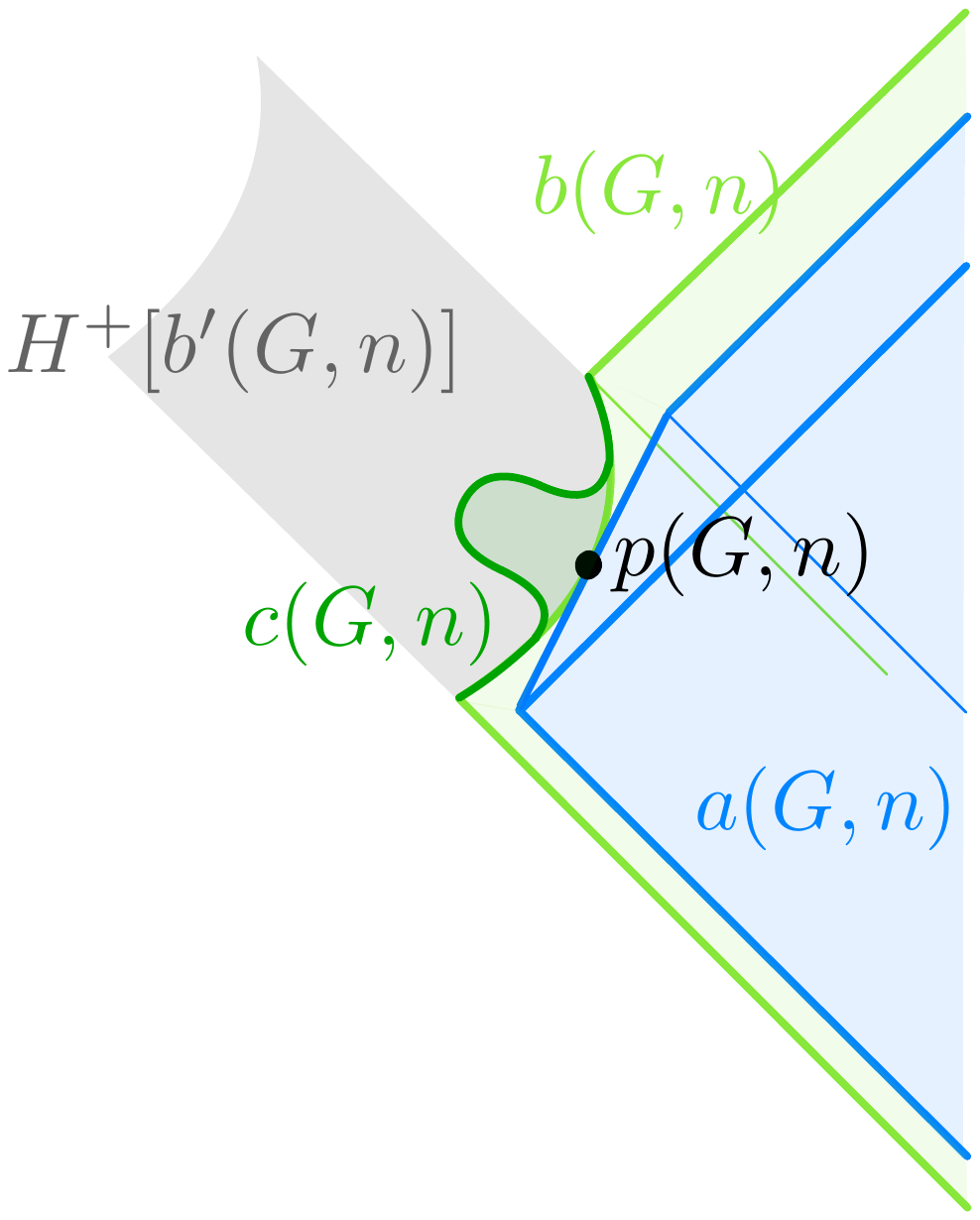}
    \captionof{figure}{Wedges appearing in the proof of the Quantum Null Energy Condition.}
    \label{fig:qnec}
\end{minipage}
\vspace{4mm}

We can also compute $\thmax^+[c(G,n),q]$ directly by splitting the expansion into its classical and leading quantum part:
\begin{multline}
    \thmax^+[c(G,n);q] -\thmax^+[b(G,n);p]  = \\
    \theta^+[c(G,n);q] -\theta^+[b(G,n);p] +4G\hbar \left( \thmax^{Q,+}[c(G,n);q] -\thmax^{Q,+}[b(G,n);p] \right)
\end{multline}
Here the quantum part is denoted with superscript $Q$ and defined as a functional derivative analogously to Eq.~\eqref{eq:thmaxdef} (see Def.~\ref{def:thmax} for details):
\begin{equation}\label{eq:thmaxqdef}
        \lim_{\lambda\to 0} \frac{1}{\lambda} H_{\rm max}(a^+_{V,\lambda}|a) = 
        4G\int_O d^{d-2} y\sqrt{h} \, V(y) \thmax^{Q,+}(y)~,
\end{equation}
The difference in classical expansions, $\theta^+[c(G,n);q] -\theta^+[b(G,n);p]$, is given by integrating the Raychaudhuri equation~\eqref{eq:raych}. Using Eq.~\eqref{eq:thmc}, and noting that $\Theta^+_{\rm max}[b(G,n);p(G,n)]$, $\theta^2$, and $\varsigma_{\mu\nu}^2$ are of order $G^2$, we can take the limit $G\to 0$ at fixed $n$ to obtain
\begin{equation}\label{eq:thmaxtmunu}
    \frac{1}{n} \left( \thmax^{Q,+}[c;q]-\thmax^{Q,+}[a;p] \right) \leq 
    \frac{2\pi}{n\hbar}\int_p^q d\lambda\, \braket{T_{\mu\nu}}_{\psi^{\otimes n}} k^\mu k^\nu
    ~.
\end{equation}
Here $\lambda$ is an affine parameter along the null generator connecting $p$ and $q$; $T_{\mu\nu}$ denotes the expectation value of the stress tensor; $k^\mu=d/d\lambda$. Recall that $c$ is the limit of $c(G;n)$ as $G\to 0$ at fixed $n$. We have divided both sides by $n$ for later convenience.

Now take the limit of this entire argument as $n\to \infty$. (This will involve choosing $G$ to scale like $1/n$ before $G\to 0$ is taken.) In this limit, the quantum equipartition principle\footnote{Here we are using only the established result for ordinary quantum systems, not the extension to generalized entropies (Conjecture~\ref{conj:aep}).} dictates that all conditional max-entropies approach conditional von Neumann entropies. In particular, $\thmax^{Q,+}[c;q]$ and $\thmax^{Q,+}[a;q]$ will be given by functional derivatives of the von Neumann entropy. Moreover, both the von Neumann expansions and the stress tensor in Eq.~\eqref{eq:thmaxtmunu} are proportional to $n$. Taking the limit as $n\to \infty$, we recover the QNEC:
\begin{equation}\label{eq:intqnec}
    \Theta_{\rm vN}^{Q,+}[c;q]-\Theta_{\rm vN}^{Q,+}[a;p] \leq 
    \frac{2\pi}{\hbar}\int_p^q d\lambda\, \braket{T_{\mu\nu}}_{\psi} k^\mu k^\nu~.
\end{equation}

This is an integrated form of the QNEC. The original form can be recovered by allowing $c$ to approach $a$. This constrains the expectation value of the stress tensor in terms of a second functional derivative of the von Neumann entropy; schematically, $\hbar S''\leq 2\pi \braket{T_{\mu\nu}} k^\mu k^\nu$~\cite{Bousso:2015mna}.

\section{Discrete-Max-Nonexpanding Wedges and Their Properties}
\label{sec:nonexpandingwedges}

In this section, we re-build the framework of normal and antinormal wedges, using only Discrete Max-Nonexpansion and Discrete Max-Focusing. We also define ``marginal wedge'' and ``throat'' to make contact with the traditional notions of stably marginal surfaces and throat-type quantum extremal surfaces. We have tried to unify the nomenclature that already exists for various special cases; see Table~\ref{tab:nomenclature}.

\begin{table}[H]
\begin{tabular}{ m{.34\textwidth} | m{.24\textwidth}| m{.34\textwidth} } 
\textbf{Definition} & \textbf{Name} & \textbf{Name if I, II, III hold} \\  
\hline\hline
$k$ is past- and future-nonexpanding at $p\in \eth k$ & $k$ is \emph{antinormal at $p$} & $k$ is \emph{accessible from $a$}\tablefootnote{This terminology was introduced in Ref.~\cite{Akers:2023fqr} for the special case where $a$ lies on the conformal boundary of Anti-de Sitter space.} \\ 
\hline
$k$ is future- (past-) nonexpanding but no past (future) outward deformation of $k$ is future- (past-) nonexpanding &$k$ is \emph{future-marginal (past-marginal)}\tablefootnote{In the classical limit and when $k$ includes an asymptotic region, $\eth k$ is said to be \emph{stably marginal}, where ``marginal'' refers to the vanishing of one of the null expansions, and ``stably'' refers to the absence of (in our language) any antinormal outward deformation of $k$.  Here we omit the word ``stably'' because there is no other kind of ``marginal.'' In our discrete setting, when numerical values of the expansions are not invoked, the absence of antinormal outward deformations is the only natural way in which we are able define the notion of marginal.} &  $k$ is \emph{future-marginally accessible from $a$ (past-marginally accessible from $a$)}\tablefootnote{In the special case where a numerical expansion exists and $k$ contains an asymptotic region, $\eth k$ was called future- or past-minimar in Refs.~\cite{Engelhardt:2018kcs,Bousso:2019dxk}.} \\ 
\hline
$k$ is future- and past-marginal & $k$ is a \emph{throat}\tablefootnote{$\eth k$ corresponds to a special case of a quantum extremal surface; it is minimal spatially and maximal in time. The other types of quantum extremal surfaces (bulge and bounce)~\cite{Engelhardt:2023bpv} do not have natural discrete definitions.} & $k$ is a \emph{throat accessible from $a$}\tablefootnote{$\eth k$ was called outer-minimal in Ref.~\cite{Engelhardt:2023bpv}.}
\end{tabular}
\vspace{3mm}
\caption{Summary of nomenclature for a wedge $k$ satisfying certain properties. The second column pertains to the definitions in Sections~\ref{sec:antinormal} and \ref{sec:stablymarginal}. The third column pertains to Sec.~\ref{sec:minimizing}, where $k$ is required to satisfies additional properties I, II, III relative to an ``input'' wedge $a$. The footnotes explain the relation to existing concepts and names.}
\label{tab:nomenclature}
\end{table}

\vspace{-3mm}
\subsection{Antinormal Wedges}
\label{sec:antinormal}

\begin{defn}[Antinormal]\label{def:antinormal}
    Let $a$ be a wedge and let $p\in\eth a$. $a$ is said to be antinormal at $p$ if $a$ is both past- and future-nonexpanding at $p$. If this holds for all $p\in \eth a$ then $a$ is called antinormal.
\end{defn}

\begin{defn}[Normal]\label{def:normal}
    Let $a$ be a wedge and let $p\in\eth a$. $a$ is said to be normal at $p$ if its complement wedge $a'$ is antinormal at $p$. If this holds for all $p\in \eth a$ then $a$ is called normal.
\end{defn}

\begin{thm}\label{thm:wua}
    The union of two antinormal wedges is antinormal.
\end{thm}

\begin{proof}
    Let $a,b$ be antinormal wedges, let $c=a\Cup b$, and let $\tilde c$ be a future-outward deformation of $c$ along $H^+(c')$. We will show that $\hmg(\tilde c|c)\leq 0$. Since an analogous argument applies to past-outward deformations of $c$, this will establish that $c$ is antinormal. 
    
    To gain some intuition, we note that by Def.~\ref{def:wedgeunion}, the edge of $c$ can be written as the following disjoint union:
    \begin{align}
        \eth c & \subset \eth a\cap b' \sqcup H^+(a')\cap H^-(b') \nonumber
        \\ & \,\sqcup \eth b\cap a' \sqcup H^-(a')\cap H^+(b')\nonumber
        \\ & \,\sqcup \eth a\cap \eth b~.
    \end{align}
    We also note that $\eth \tilde c\subset J^+(\eth c)$, and that the portions of $\eth \tilde c$ that lie in the causal future of the first, second, and third line sets lie, respectively, on $H^+(a')$, on $H^+(b')$, and on $H^+(a)\cap H^+(b)$.
    
    This motivates us to apply the QFC to two wedges $\tilde c_a$ and $c_a$ that both lie in $I^+(a)$ and are obtained by outward deformations of $a$ along $H^+(a')$ to the edges of $\tilde c$ and $c$: 
    \begin{align}\label{eq:tildeca}
        \tilde c_a & = (a\cup (\eth \tilde c\cap H^+(a')))'' ~; \\
        c_a & = (a\cup (\eth c\cap H^+(a')))'' ~.
    \end{align}
    Since $a$ is antinormal, 
    \begin{equation}\label{eq:uan1}
        \hmg(\tilde c_a|c_a)\leq 0~.
    \end{equation}
    Similarly, we define 
    \begin{align}
        \tilde c_b & = (b\cup (\eth \tilde c\cap H^+(b')))'' ~; \\
        c_b & = (b\cup (\eth c\cap H^+(b')))'' ~.
    \end{align}
    Since $b$ is antinormal, 
    \begin{equation}\label{eq:uan2}
        \hmg(\tilde c_b|c_b)\leq 0~.
    \end{equation}
    By Eqs.~\eqref{eq:uan1}, \eqref{eq:uan2}, and Discrete Subadditivity (Conj.~\ref{conj:dsa}), 
    \begin{align}\label{eq:uan}
        \hmg(\tilde c_a\Cup c_b|c_a\Cup c_b) &\leq 0~; \\
        \hmg(\tilde c_b\Cup \tilde c_a|c_b\Cup \tilde c_a) &\leq 0~.
    \end{align}
    We now apply the chain rule to obtain
    \begin{align}
        \hmg(\tilde c|c) 
        & = \hmg(\tilde c_a\Cup \tilde c_b|c_a\Cup c_b)  \\
        & \leq \hmg(\tilde c_a\Cup \tilde c_b|\tilde c_a\Cup c_b) +\hmg(\tilde c_a\Cup c_b|c_a\Cup c_b)\leq 0~.
    \end{align}
\end{proof}

\begin{cor}\label{lem:win}
    The intersection of two normal wedges is normal.
\end{cor}

\begin{proof}
    The result follows immediately from Definitions~\ref{def:wedgeunion} and~\ref{def:normal}, and the preceding Theorem.
\end{proof}

\begin{cor}\label{cor:wua}
    Let $a$, $b$, and $h$ be wedges such that $h\subset a\cap b$. Suppose that $a$ is antinormal at all points on $\eth a\setminus \eth h$, and that $b$ is antinormal at all points on $\eth b\setminus \eth h$. Then $c=a\Cup b$ is antinormal at all points on $\eth c\setminus \eth h$.
\end{cor}

\begin{proof}
    The proof is a straightforward generalization of the proof of the preceding Theorem. Instead of Eq.~\eqref{eq:tildeca}, we now define $\tilde c_a = (a\cup (\eth \tilde c\cap H^+(a')\setminus H^+(h')))''$, and similarly for $c_a$, $\tilde c_b$, and $c_b$. This ensures that only antinormal pieces of $\eth a$ and $\eth b$ are used when the QFC is invoked. It remains true that $c=c_a\Cup c_b$ and $\tilde c = \tilde c_a\Cup \tilde c_b$ since any portion of $\eth c$ or $\eth \tilde c$ that lies on $H^+(h')$ must also lie on $\eth a$ or on $\eth b$ itself. 
\end{proof}

\subsection{Marginal Wedges and Throats}
\label{sec:stablymarginal}

\begin{defn}[Marginal] \label{def:stablymarginal}
    A wedge $a$ is said to be future-marginal if 
    \begin{enumerate}[I.]
        \item $a$ is antinormal;
        \item $a$ is future-noncontracting.
    \end{enumerate}
    Past-marginal is defined analogously.
\end{defn}
\begin{defn}[Throat]\label{def:throat}
    A wedge $a$ is said to be a throat if $a$ is future- and past-marginal.
\end{defn}

\begin{thm}\label{thm:fmtopastmarginal}
    If the wedge $a$ is future-marginal, and if $\thmax^+(a)$ exists, then $\thmax^+(a)=0$.
\end{thm}

\begin{proof}
    Since $a$ is future-nonexpanding, Theorem~\ref{thm:fnetononpos} implies that $\thmax^+(a)\leq 0$. Suppose for contradiction that $\thmax^+(a)< 0$. By Def.~\ref{def:thmax}, $\thmax^+$ is continuous on its open domain, so there exists an open neighborhood $O$ of $\eth a$ in $M$ such that $\thmax^+(a)<0$ for all $q\in O\cap\eth a$. Assuming continuity of $\thmax^+(a)$, we can smoothly deform $a$ outward along $H^+(a')\cap O$, to obtain a wedge $b\supset a$ with $\thmax^+(b;r)<0$ for all  $r\in \eth b \setminus \eth a$. By Theorem~\ref{thm:negtofne}, $b$ is a future-nonexpanding past outward deformation of $a$. This contradicts the assumption that $a$ is future-marginal.
\end{proof}

\begin{cor}\label{cor:throattoextremal}
    If the wedge $a$ is a throat, and if $\thmax^\pm(a)$ both exist, then $\thmax^\pm(a)=0$.
\end{cor}

\begin{proof}
    The result follows from the preceding theorem.
\end{proof}

\subsection{Accessible Wedges and Accessible Throats}
\label{sec:minimizing}

\begin{defn}[Accessible Wedge]\label{def:minimizing}
    Let $a$ be a wedge. The wedge $f$ is called accessible from $a$ if it satisfies the following three properties:
    \begin{enumerate}[I.]
    \item $f\supset a$, $\tilde \eth f=\tilde \eth a$;\footnote{This condition will be updated in future work, for generality. Here as in Ref.~\cite{Bousso:2023sya}, we are mainly interested in the cases where $a$ (even if it shares no boundary with conformal infinity) is a wedge in an asymptotically AdS spacetime. The condition $\tilde \eth f=\tilde \eth a$ is not suitable when the spacetime has a different asymptotic structure. A more general condition will be supplied in forthcoming work~\cite{BoussoKaya}, alongside an alternate definition of $\emin$ as a complement of a suitable $\emax$.}
    \item $f$ is antinormal at points $p\in \eth f\setminus\eth a$;
    \item $f$ admits a Cauchy slice $\Sigma$ such that $\Sigma\supset \eth a$ and such that for any wedge $h\subsetneq f$ with $a\subset h$, $\eth h\subset \Sigma$, and $\eth h\setminus \eth f$ compact in $M$,
    \begin{equation}
        \hmg(f|h)< 0~.
    \end{equation}
    \end{enumerate}
\end{defn}

\begin{rem}
    Our notion of ``accessible'' is a generalization to arbitrary input regions $a$, of the notion of ``max-accessible'' introduced in Ref.~\cite{Akers:2023fqr}. 
\end{rem}

\begin{defn}[Marginally accessible Wedge]\label{def:minimar}
    Let $a$ be a wedge. A wedge $k$ is called future-marginally accessible from $a$ if
    \begin{enumerate}[I.]
        \item $k$ is accessible from $a$;
        \item $k$ is future-noncontracting.
    \end{enumerate}
    Past-marginally accessible is defined analogously.
\end{defn}

\begin{rem}
    In the limit where $a$ is an asymptotic region near the boundary of Anti-de Sitter space, the above definition reduces to the ``minimar'' surfaces defined in Ref.~\cite{Engelhardt:2018kcs} (in the classical limit) and in Ref.~\cite{Bousso:2019dxk} (for the von Neumann generalized entropy).
\end{rem}

\begin{defn}[Accessible Throat]\label{def:minithroat}
    Let $a$ be a wedge. A wedge $k$ is called a throat accessible from $a$ if $k$ is past- and future-marginally accessible from $a$. (See the previous definition and Def.~\ref{def:throat}.)
\end{defn}

\begin{rem}\label{rem:outerminimizing}
    In the limit where $a$ is an asymptotic region near the boundary of Anti-de Sitter space and where we approximate the conditional max entropy by a difference of von Neumann entropies, the the above definition reduces to the notion of an ``outer-minimizing'' wedge for a given asymptotic boundary subregion. The word ``outer'' here refers to the side of the throat quantum extremal surface $\eth k$ on which property III above is met. Here we take the wedge $k$ as the more fundamental object, from which $\eth k$ arises as its edge; so the terminology ``outer'' is obsolete.
\end{rem}

\section{Generalized Entanglement Wedges and Their Properties}
\label{sec:entanglement}

We will now show that the ``bare-bones'' structure developed in this paper suffices to define generalized entanglement wedges and to establish their properties. We will make no reference to numerical values of the expansions, nor to the von Neumann and (outward) min-expansions. Only the Discrete Max-Expansion is used in definitions, and only Discrete Max-Focusing is used in proofs. Aside from this main objective, this section also serves to combine the results of Ref.~\cite{Bousso:2023sya} (which approximated matter states as compressible) with those of Refs.~\cite{Akers:2023fqr} (which did not apply to arbitrary spacetimes).

\begin{defn}[Max entanglement Wedge of a Gravitating Region]\label{emaxdef}
    Given a wedge $a$, let $F(a)$ 
    be the set of all wedges accessible from $a$ (see Def.~\ref{def:minimizing}). The {\em max entanglement wedge} of $a$, $\emax(a)$, is their wedge union:
   \begin{equation}\label{eq:emaxdef}
       \emax(a) \equiv \Cup_{f\in F(a)}\, f~.
   \end{equation}
\end{defn}
\begin{thm}\label{lem:emaxpropshare}
    $\emax(a)$ is accessible from $a$.
\end{thm}

\begin{proof} 
We must show that $\emax(a)\in F(a)$, i.e., that $\emax(a)$ satisfies properties I--III listed in Def.~\ref{def:minimizing}.

\emph{Property I:} $f=a$ satisfies properties I--III with any choice of Cauchy slice. Hence $F(a)$ is nonempty, and Eq.~\eqref{eq:emaxdef} implies that $f=\emax(a)$ satisfies property I.

\emph{Property II:} Proceeding inductively, let $f_1,f_2\in F(a)$ and $f_3=f_1\Cup f_2$. By Corollary~\ref{cor:wua}, $f_3$ is antinormal at points $p \in \eth f_3 \setminus \eth a$. By induction, $\emax(a)$ is antinormal at points $p \in \eth \emax(a) \setminus \eth a$.

\emph{Property III:} Again proceeding inductively, let $f_1,f_2\in F(a)$, with property III satisfied by Cauchy slices $\Sigma_1$ and $\Sigma_2$, respectively. Then $f_3=f_1\Cup f_2$ admits the Cauchy slice
\begin{equation}\label{eq:unionslice}
  \Sigma_3 = \Sigma_1 \cup  [H^+(f_1')\cap \mathbf{J}^-(\Sigma_2)] 
  \cup [H^-(f_1')\cap \mathbf{J}^+(\Sigma_2)] \cup [\Sigma_2 \cap f_1']~.
\end{equation}
Let $h\supset a$, $\eth h\subset \Sigma_3$. By property III of $\Sigma_1$ and $\Sigma_2$, 
\begin{align}
    \hmg[\Sigma_1|h\cap\Sigma_1] &< 0~, \\
    \hmg[\Sigma_2|h\cap\Sigma_2] & < 0~.
\end{align}
Then Discrete Subadditivity (Conj.~\ref{conj:dsa}) implies
\begin{align}
    \hmg[h\Cup \Sigma_1|h] &< 0~, \\
    \hmg[\Sigma_3|h \Cup (\Sigma_3\setminus \Sigma_2)] &< 0~.
\end{align}
By Discrete Max-Focusing (Conj.~\ref{conj:qfc}), 
\begin{equation}
    \hmg[h \Cup (\Sigma_3\setminus \Sigma_2)|h\Cup \Sigma_1] \leq 0~.
\end{equation}
Finally, by the chain rule Eq.~\eqref{eq:chainMaxMaxMax},
\begin{align}
    \hmg[\Sigma_3|h] &\leq \hmg[\Sigma_3|h \Cup (\Sigma_3\setminus \Sigma_2)] \\
    &\qquad + \hmg[h \Cup (\Sigma_3\setminus \Sigma_2)|h\Cup \Sigma_1] + \hmg[h\Cup \Sigma_1|h] \nonumber\\
    &< 0\nonumber~.
\end{align}
Hence $f_3$ satisfies property III.
\end{proof}

\begin{thm}\label{thm:emaxextremal}
    $\emax(a)$ is an accessible throat.
\end{thm}
\begin{proof}
We will show that $\emax(a)$ is future-marginally accessible; the proof that it is past-marginally accessible is analogous. Suppose for contradiction that $\emax(a)$ is antinormal on $\eth \emax(a)\setminus\eth a$ but not future-marginal. Then for any open set $O\supset\eth \emax(a)$, there exists a future-nonexpanding past-directed outward deformation $f_1$ of $\emax(a)$ such that $\eth f_1\subset\eth\emax(a)\cup [O\cap H^-(\emax(a)')]$. Since $\emax(a)$ is antinormal on the open set $\eth\emax(a)\setminus\eth a$, $O$ can be chosen such that $\emax(a)$ is antinormal on $O\cap\eth\emax(a)$. Then by Discrete Max-Focusing, $f_1$ is also past-nonexpanding. Hence $f_1$ satisfies properties I and II of Def.~\ref{def:minimizing}. We will now show that $f_1$ satisfies property III with Cauchy surface $\Sigma_1 = \smax\cup H^-(\emax(a)')\cap f_1$, where $\smax$ is the Cauchy surface of $\emax(a)$ on which property III is satisfied. 

Let $h\subset f_1$ with $\eth h\subset \Sigma_1$, $h\supset a$. By Property III of $\emax(a)$, 
\begin{equation}
    \hmg[\emax(a)|h\cap \emax(a)]<0~,
\end{equation} 
By Discrete Subadditivity, Conj.~\ref{conj:dsa}:
\begin{equation}
    \hmg[h\Cup \emax(a)|h]<0~.
\end{equation}
By Discrete Max-Focusing, Conj.~\ref{conj:qfc},
\begin{equation}
    \hmg[f_1\Cup \emax(a)|h\Cup \emax(a)] \leq 0~,
\end{equation}
By the chain rule Eq.~\eqref{eq:chainMaxMaxMax}, the previous two equations imply
\begin{equation}
    \hmg[f_1\Cup \emax(a)|h] < 0~.
\end{equation}
Since $f_1 \not\subset \emax(a)$, the fact that it satisfies all three properties contradicts the definition of $\emax(a)$.

The second part of the theorem follows since $\emax(a)$ is antinormal on $\eth\emax(a)\setminus\eth a$, by Theorem~\ref{lem:emaxpropshare}.
\end{proof}

\begin{defn}[min entanglement Wedge of a Gravitating Region]\label{emindef}
    Given a wedge $a$, let $G(a)\equiv \set{g: \mathrm{i}\, \wedge\, \mathrm{ii} \,\wedge\, \mathrm{iii}}$ be the set of all wedges that satisfy the following properties:
    \begin{enumerate}[i.]
    \item $g\supset a$;
    \item $g$ is normal; 
    \item $g'$ admits a Cauchy slice $\Sigma'$ such that for any wedge $h\supsetneq g$ with $\eth h\subset \Sigma'$ and $\eth h\setminus \eth g$ compact in $M$,
    \begin{equation}
        \hmg(g'|h')< 0~.
    \end{equation}
    \end{enumerate}
    The {\em smooth conditional min entanglement wedge} of $a$, $\emin(a)$, is their intersection:
    \begin{equation}\label{eq:emindef}
      \emin(a)\equiv \cap_{g\in G(a)}\, g~.
    \end{equation}
\end{defn}

\begin{thm}\label{lem:eminpropshare}
$\emin(a)\in G(a)$. 
\end{thm}

\begin{proof} 
We must show that $\emin(a)$ satisfies properties i--iii listed in Def.~\ref{emindef}.

\emph{Property i:} $g=M$ trivially satisfies properties i--iii, so $G(a)$ is nonempty. Property i then implies $\emin(a)\supset a$.

\emph{Property ii:} The intersection of two normal wedges is normal by Lemma~\ref{lem:win}. Hence $\emin(a)$ is normal by Def.~\ref{emindef}.

\emph{Property iii:} Let $g_1,g_2\in G(a)$ with property iii satisfied by Cauchy slices $\Sigma'_1$ and $\Sigma'_2$, respectively; and let $g_3=g_1\cap g_2$. Then $g_3'$ admits the Cauchy slice
\begin{equation}\label{eq:sminunion}
  \Sigma'_3 = \Sigma'_1 \cup  [H^+(g_1)\cap \mathbf{J}^-(\Sigma'_2)] 
  \cup [H^-(g_1)\cap \mathbf{J}^+(\Sigma'_2)] \cup [\Sigma'_2 \cap g_1]~.
\end{equation}
Let $h\supset g_3$, $\eth h\subset \Sigma'_3$. By property iii of $\Sigma'_1$ and $\Sigma'_2$, 
\begin{align}
    \hmg[\Sigma_1'|h'\cap\Sigma_1'] & < 0~,\\
    \hmg[\Sigma'_2|h'\cap\Sigma'_2] & < 0~.
\end{align}
Then Discrete Subadditivity (Conj.~\ref{conj:dsa}) implies
\begin{align}
    \hmg[h'\Cup \Sigma_1'|h'] & < 0~, \\
    \hmg[\Sigma'_3|h' \Cup (\Sigma_3'\setminus \Sigma_2')] & < 0~.
\end{align}
By Discrete Max-Focusing (Conj.~\ref{conj:qfc}),
\begin{equation}
    \hmg[h' \Cup (\Sigma_3'\setminus \Sigma_2')|h'\Cup \Sigma_1'] \leq 0~.
\end{equation}
Summing the above three inequalities and using the chain rule Eq.~\eqref{eq:chainMaxMaxMax}, we obtain $\hmg[\Sigma'_3|h']< 0$.
Hence $g_3$ satisfies property iii.
\end{proof}

\begin{thm}
    $\emax (a) \subset\emin(a)$
\end{thm}
\begin{proof}
    Since only one input wedge $a$ is involved, we suppress the argument of $\emax$ and $\emin$. We will first consider two special cases for pedagogical purposes, then we treat the general case.
    
    Special case 1: Suppose that $\smax$ and $\smin$ lie on a common Cauchy slice. By property III, 
    \begin{equation}
        0 > \hmg[\emax | (\emin \cap \smax)]~.
    \end{equation}
    Then by Discrete Subadditivity (Conj.~\ref{conj:dsa}), Duality (Prop 3.12 of \cite{Akers:2023fqr}), Theorem 2.32 of \cite{Akers:2023fqr}, and property iii, respectively,
    \begin{align}
        0 &> \hmg[(\emax' \cap \smin)'|\emin] \\
        &= -\hmingen[\emin'|\emax' \cap \smin] \\
        &\geq -\hmg[\emin'|\emax' \cap \smin] \\
        &> 0~.
    \end{align}
    
    Special case 2: Suppose that $\smax\subset J(\smin)$, i.e., $\smax$ and $\smin$ are everywhere causally separated. By Discrete Max-Focusing (Conj.~\ref{conj:qfc}),
    \begin{align}
        \hmg[(\emax' \cap \smin)'|\emax] &\leq 0~, \\
        \hmg[(\emin \cap \smax)'|\emin'] &\leq 0~.
    \end{align}
    Then,
    \begin{align}
        0 &>\hmg[(\emax{}' \cap \smin)'|\emax] + \hmg[\emax|\emin \cap \smax] \\
        &\geq \hmg[(\emax' \cap \smin)'|\emin \cap \smax] \\
        &= -\hmingen[(\emin \cap \smax)'|\emax' \cap \smin] \\
        &\geq -\hmg[(\emin \cap \smax)'|\emax' \cap \smin] \\
        &\geq -\hmg[(\emin \cap \smax)'|\emin'] - \hmg[\emin'|\emax' \cap \smin] \\
        &> 0~,
    \end{align}
    where the second and fifth lines follow from the chain rule, and the first and sixth line follow from a combination of QFC and properties III and iii.
    
    General proof: we note that
    \begin{equation}
        ((\emax' \cap \smin)\Cup(\emax'\cap\emin))' = (\emax' \cap \smin)'\cap(\emax'\cap\emin)'~.
    \end{equation}
    By Discrete Max-Focusing (Conj.~\ref{conj:qfc}),
    \begin{align}
        \hmg[(\emax' \cap \smin)'\cap(\emax'\cap\emin)'|\emax] &\leq 0 \\
        \hmg[(\emin \cap \smax)'\cap(\emin\cap\emax')'|\emin'] &\leq 0~.
    \end{align}
    Then
    \begin{align*}
        0 &>\hmg[(\emax' \cap \smin)'\cap(\emax'\cap\emin)'|\emax] + \hmg[\emax|\emin \cap \smax] \\
        &\geq \hmg[(\emax' \cap \smin)'\cap(\emax'\cap\emin)'|\emin \cap \smax] \\
        &= -\hmingen[(\emin \cap \smax)'|((\emax' \cap \smin)'\cap(\emax'\cap\emin)')'] \\ 
        &\geq -\hmg[(\emin \cap \smax)'|(\emax' \cap \smin)\Cup(\emax'\cap\emin)]~.
    \end{align*}
    By Discrete Subadditivity (Conj.~\ref{conj:dsa}),
    \begin{align*}
        0&> -\hmg[(\emin \cap \smax)'\cap(\emin\cap\emax')'|\emax' \cap \smin] \\
        &\geq -\hmg[(\emin \cap \smax)'\cap(\emin\cap\emax')'|\emin'] -\hmg[\emin'|\emax' \cap \smin] \\
        &\geq 0~.
    \end{align*}
    This is a contradiction, unless $\emin \cap \smax = \emax$ and $(\emax' \cap \smin)'=\emin$. Hence $\emax\subset\emin$.
\end{proof}

\begin{thm}[Nesting of $\emin$]\label{thm:nesting}
For wedges $a$ and $b$,
\begin{equation}
    a\subset b \implies \emin(a)\subset \emin(b)~.
\end{equation}
Moreover, $\smin(a)$ can be chosen so that
\begin{equation}
    \smin(a)\supset\smin(b)~.
\end{equation}
\end{thm}

\begin{proof}
    The proof is identical to the proof of Theorem 27 in  Ref.~\cite{Bousso:2023sya}.
\end{proof}

\begin{cor}
    If $a\subset b$ and $\eth \emin(b)\setminus \emin(a)$ is compact, then 
    \begin{equation}
        \hmg[\emin(a)'|\emin(b)']\leq 0~.
    \end{equation}
\end{cor}

\begin{proof}
    The proof is identical to the proof of Corollary 28 in  Ref.~\cite{Bousso:2023sya}.
\end{proof}

\begin{thm}[No Cloning]\label{thm:nocloning}
   \begin{equation}
     a \subset \emin'(b)~\mbox{and}~b \subset \emax'(a)~\implies~
     \emax(a)\subset\emin'(b)~.
  \end{equation}
\end{thm}

\begin{proof}
Let 
\begin{equation}
    g=\emin(b)\cap \emax'(a)~.
\end{equation}
We will show that $g$ satisfies properties i-iii listed in Def.~\ref{emindef}. This contradicts the definition of $\emin(b)$ unless $g=\emin(b)$, which is equivalent to the conclusion.

Property i: By assumption, $b\subset \emax'(a)$. By Theorem~\ref{lem:eminpropshare}, $b\subset \emin(b)$. Hence $b\subset g$.

Property ii: By Theorem~\ref{lem:eminpropshare}, $\emin(b)$ is normal, and by Theorem~\ref{lem:emaxpropshare}, $\emax(a)$ is antinormal. Then $\emax(a)'$ is normal, and by Lemma~\ref{lem:win}, $g$ is normal. 

Property iii: Let 
\begin{multline}
    \Sigma' = \smin(b) 
    \cup \left( H^+[\emin(b)] \cap \mathbf{J}^-[\smax(a)]\right) \\
    \cup \left( H^-[\emin(b)] \cap \mathbf{J}^+[\smax(a)]\right)
    \cup [\smax(a) \cap \emin(b)]~.
\end{multline}
This is a Cauchy slice of $g'$. Let $h\supset g$ with $\eth h\subset \Sigma'$. By property iii of $\smin(b)$ and property III of $\smax(a)$
\begin{align}
    \hmg[\smin(b)|h'\cap\smin(b)] & < 0~,\\
    \hmg[\smax(a)|h'\cap\smax] & < 0~.
\end{align}
Then by Discrete Subadditivity (Conj.~\ref{conj:dsa}),
\begin{align}
    \hmg[h'\Cup \smin(b)|h'] & < 0~,\\
    \hmg[\Sigma'|h' \Cup (\Sigma'\setminus \smax(a))] & < 0~.
\end{align}
By Discrete Max-Focusing (Conj.~\ref{conj:qfc}),
\begin{equation}
    \hmg[h' \Cup (\Sigma'\setminus \smax(a))|h'\Cup \smin] \leq 0~.
\end{equation}
Finally, by the chain rule Eq.~\eqref{eq:chainMaxMaxMax}, we see that
\begin{align*}
    \hmg[\Sigma'|h'] &\leq \hmg[h'\Cup \smin|h'] + \hmg[h' \Cup (\Sigma'\setminus \smax(a))|h'\Cup \smin] \\
    &\qquad + \hmg[\Sigma'|h' \Cup (\Sigma'\setminus \smax(a))] \\
    &< 0 ~.
\end{align*}
Hence $g$ satisfies property iii.
\end{proof}

\begin{rem}
The following theorem provides powerful evidence that when $\emin(a)=\emax(a)$, the semiclassical state of the entanglement wedge is the same, at the level of the density operator, as the fundamental state of the input region $a$ (or in the special case of AdS/CFT, the state of a CFT subregion). Otherwise there would be no reason for entanglement wedges to obey strong subadditivity. In particular, strong subadditivity of the entanglement wedges is not a trivial consequence of Conjecture~\ref{conj:ssa}, since the wedges appearing on the right hand side of the following theorem need not be related to the left hand side wedges by union with a spacelike separated wedge.
\end{rem}

\begin{thm}[Strong Subadditivity of the Generalized Max and Min Entropies of Entanglement Wedges]
\label{thm:ssa}
Suppressing $\Cup$ symbols where they are obvious, let $a$, $b$, and $c$ be mutually spacelike wedges, such that
\begin{align}\nonumber
    \emin(ab) &=\emax(ab)~,~\emin(bc)=\emax(bc)~,\\ \nonumber
    \emin(b) &=\emax(b)~,~\mbox{and}~\emin(abc)=\emax(abc)~.
\end{align}
Then (writing $e$ for $\emin=\emax$)
\begin{align}
    \hmg[e(bc)|e(b)] &\geq \hmg[e(abc)|e(ab)]~;\label{eq:maxssa}\\
    \hmingen[e(bc)|e(b)] &\geq \hmingen[e(abc)|e(ab)]\label{eq:minssa}~.
\end{align}
\end{thm}

\begin{proof}
Following~\cite{Bousso:2023sya}, we define a wedge $x$ by the Cauchy slice of its complement:
\begin{equation}\label{eq:xdef}
    \Sigma'(x) = \Sigma'[e(ab)] \cup
    \left( H^+[e(ab)] \cap \mathbf{J}^-[\eth e(bc)]\right) \cup
    \left( H^-[e(ab)] \cap \mathbf{J}^+[\eth e(bc)]\right)~.
\end{equation}
Note that $\eth x$ is nowhere to the past or future of $\eth e(bc)$. Therefore, there exists a single Cauchy slice that contains the edges of $x$, $e(bc)$, $x\cap e(bc)$, and $x\Cup e(bc)$. We note that $x\cap e(bc)=e(ab)\cap e(bc)$ and that $x\Cup e(bc) \subset e(ab)\Cup e(bc)$. By Theorem~\ref{thm:nesting},
\begin{equation}
    x\Cup e(bc) \subset e(abc)~~\mbox{and}~~x\cap e(bc) \supset e(b)~.
\end{equation}

By Corollary~\ref{lem:win}, $x\cap e(bc)$ is normal. By Discrete Max-Focusing, property iii of $e(b)$, and the chain rule Eq.~\eqref{eq:chainMinMinMin}, respectively,
\begin{align}
    &\hmingen[x\cap e(bc)|e(b) \Cup (\Sigma'(b)\cap x\cap e(bc))]\geq0 \\
    &\hmingen[e(b) \Cup (\Sigma'(b)\cap x\cap e(bc))|e(b)]\geq 0 \\
    &\hmingen[x\cap e(bc)|e(b)]\geq 0~,\label{eq:SSAproof min chain}
\end{align}
where we remark that we didn't use a strict inequality as in property iii to account for the possibility that $\Sigma'(b)\cap x\cap e(bc)=e(b)$.

Although $x$ need not be antinormal, Eq.~\eqref{eq:xdef} ensures that $[\eth x\setminus \eth e(ab)]\cap\eth(x\Cup e(bc))=\varnothing$. Arguments analogous to the proof of Corollary~\ref{cor:wua} then imply that $x\Cup e(bc)$ is antinormal, except where its edge coincides with $\eth[e(abc)]$ and hence with $\Sigma[e(abc)]$. By Discrete Max-Focusing, property III of $e(abc)$, and the chain rule Eq.~\eqref{eq:chainMaxMaxMax},
\begin{align}
    &\hmg[\Sigma(abc)\setminus (x' \cap e(bc)')|x\Cup e(bc)]\leq0 \\
    &\hmg[e(abc)|\Sigma(abc)\setminus (x' \cap e(bc)')] \leq 0 \\
    &\hmg[e(abc)|x\Cup e(bc)]\leq 0~. \label{eq:SSAproof max chain}
\end{align}
Normality of $e(ab)$ and the QFC imply $\hmg[x'|e(ab)']\leq 0$, i.e.,
\begin{equation}\label{eq:SSAproof qfc}
    \hmingen[e(ab)|x]\geq 0~.
\end{equation}

We now prove Eq.~\eqref{eq:maxssa} as follows:
\begin{align}
    \hmg[e(bc)|e(b)] &\geq \hmg[e(bc)|x\cap e(bc)] + \hmingen[x\cap e(bc)|e(b)] \\
    & \geq \hmg[e(bc)|x\cap e(bc)] \\
    &\geq \hmg[x\Cup e(bc)|x] \\
    &\geq \hmg[e(abc)|x\Cup e(bc)] + \hmg[x\Cup e(bc)|x] \\
    &\geq \hmg[e(abc)|x] \\
    &\geq \hmg[e(abc)|e(ab)] + \hmingen[e(ab)|x] \\
    &\geq \hmg[e(abc)|e(ab)]~,
\end{align}
where we applied the chain rule Eq.~\eqref{eq:chainMaxMaxMin} in the first inequality; Eq.~\eqref{eq:SSAproof min chain} in the second; strong subadditivity of the generalized max entropy, Conjecture~\ref{conj:ssa}, in the third; Eq.~\eqref{eq:SSAproof max chain} in the fourth; the chain rule Eq.~\eqref{eq:chainMaxMaxMax} in the fifth; the chain rule Eq.~\eqref{eq:chainMaxMaxMin} in the sixth, and Eq.~\eqref{eq:SSAproof qfc} in the final inequality.

The proof of Eq.~\eqref{eq:minssa} proceeds similarly:
\begin{align}
    \hmingen[e(abc)|e(ab)] &\leq \hmingen[e(abc)|e(ab)] + \hmingen[e(ab)|x]\\
    &\leq \hmingen[e(abc)|x]\\
    &\leq \hmg[e(abc)|x\Cup e(bc)] + \hmingen[x\Cup e(bc)|x] \\
    &\leq \hmingen[x\Cup e(bc)|x] \\
    &\leq \hmingen[e(bc)|x\cap e(bc)] \\
    &\leq \hmingen[e(bc)|x\cap e(bc)] + \hmingen[x\cap e(bc)|e(b)] \\
    &\leq \hmingen[e(bc)|e(b)]~,
\end{align}
where we applied Eq.~\eqref{eq:SSAproof qfc} in the first inequality; the chain rule Eq.~\eqref{eq:chainMinMinMin} in the second; the chain rule Eq.~\eqref{eq:chainMinMaxMin} in the third; Eq.~\eqref{eq:SSAproof max chain} in the fourth; strong subadditivity of the generalized min entropy, Remark~\ref{rem:ssamin}, in the fifth; Eq.~\eqref{eq:SSAproof min chain} in the sixth; and the chain rule Eq.~\eqref{eq:chainMinMinMin} in the final inequality.
\end{proof}

\begin{cor}
    One can eliminate the assumption that $\emin(b)=\emax(b)$, so long as $e(b)$ is replaced by $\emin(b)$ in the conclusions of the theorem. This is a more general result; however, its interpretation is less clear to us.
\end{cor}

\begin{cor}\label{cor:dsa}
    Theorem~\ref{thm:ssa} immediately implies both versions of discrete subadditivity of the generalized max and min entropies of entanglement wedges:
    \begin{enumerate}[i.]
        \item \begin{align}
        \hmg[e(bc)|e(b)]\leq0 &\quad\Rightarrow\quad \hmg[e(abc)|e(ab)] \leq0~;\\
        \hmingen[e(abc)|e(ab)]\geq0 & \quad\Rightarrow\quad \hmingen[e(bc)|e(b)] \geq0~.
        \end{align}
        \item \begin{align}
        \hmg[e(bc)|e(b)]<0 &\quad\Rightarrow\quad \hmg[e(abc)|e(ab)] <0~;\\
        \hmingen[e(abc)|e(ab)]>0 & \quad\Rightarrow\quad \hmingen[e(bc)|e(b)] >0~.
    \end{align}
    \end{enumerate}
    We note that these results can also be obtained directly by replacing Conjecture~\ref{conj:ssa} with the weaker Conjecture~\ref{conj:dsa} in the proof of the theorem.
\end{cor}

\begin{thm} [Strong Subadditivity of the von Neumann Entropies of Entanglement Wedges]
    Under the same assumptions as Theorem~\ref{thm:ssa},
    \begin{equation}\label{eq:ssavN}
        \S[e(bc)|e(b)] \geq \S[e(abc)|e(ab)]~.
    \end{equation}
\end{thm}
\begin{proof}
    Let $\psi$ be the quantum state on $M$ and let $M^{\otimes k}$ denote $k$ copies of the same quantum field theory, in the state $\psi^{\otimes k}$. By Theorem~\ref{thm:ssa},
    \begin{equation}
        \hmg\left[e[(bc)^{\otimes k}]|e[b^{\otimes k}] \right]_{M^{\otimes k}} \geq \hmg\left[e[(abc)^{\otimes k}]|e[(ab)^{\otimes k}] \right]_{M^{\otimes k}}~.
    \end{equation}
    Since $e(a)=\emax(a)=\emin(a)$, we have $e(a^{\otimes k})=e(a)^{\otimes k}$. Then
    \begin{equation}
        \hmg\left[e(bc)^{\otimes k}|e(b)^{\otimes k} \right]_{M^{\otimes k}} \geq \hmg\left[e(abc)^{\otimes k}|e(ab)^{\otimes k} \right]_{M^{\otimes k}}~,
    \end{equation}
    and Eq.~\eqref{eq:ssavN} follows from Conjecture~\ref{conj:aep}.
\end{proof}

\subsection*{Acknowledgements}
We would like to thank C.~Akers, L.~Bariuan, N.~Engelhardt, S.~Kaya, A.~Levine, G.~Penington, A.~Shahbazi-Moghaddam, and J.~Yeh for discussions. This work was supported by the Department of Energy, Office of Science, Office of High Energy Physics under QuantISED Award DE-SC0019380.

\bibliographystyle{JHEP}
\bibliography{covariant}

\end{document}